\documentclass[journal]{IEEEtran}
\usepackage[T1]{fontenc}
\usepackage{times}
\usepackage{amsmath}
\usepackage{amsthm}
\usepackage{color}
\usepackage{float}
\usepackage{array}
\usepackage{amssymb}
\usepackage{amsfonts}
\usepackage{graphicx}
\usepackage{verbatim} 
\usepackage{multirow} 
\usepackage{slashed,graphicx}
\usepackage{url}
\usepackage{tablefootnote}
\usepackage{xcolor}
\usepackage{cite}
\usepackage{tabulary}
\usepackage{booktabs}
\usepackage{mathtools}
\usepackage{textcomp}
\usepackage{algorithm}
\usepackage{algorithmic}

\usepackage{tikz}

\newcommand{\algrule}[1][.1pt]{\par\vskip.5\baselineskip\hrule height #1\par\vskip.5\baselineskip}

\newcommand{\IFTHEN}[2]{
 \STATE \algorithmicif\ #1\ \algorithmicthen\ #2\ }

\newtheorem{mydef}{Definition}
\newtheorem{mylemma}{Lemma}
\newtheorem{mytheorem}{Theorem}
\newtheorem{myproposition}{Proposition}

\ifCLASSOPTIONcompsoc \usepackage[caption=false,font=normalsize,labelfont=sf,textfont=sf,labelformat=simple]{subfig}
\else
\usepackage[caption=false,font=footnotesize,labelformat=simple]{subfig}
\fi

\ifodd 1

\newcommand{\rev}[1]{{{\color{black} #1}}} 
\newcommand{\revv}[1]{{{\color{black} #1}}} 
\newcommand{\res}[1]{{{\color{black} #1}}} 
\else
\newcommand{\rev}[1]{#1}
\newcommand{\revv}[1]{#1}
\newcommand{\com}[1]{}
\fi

\begin{document}
\title{SAS-Assisted Coexistence-Aware Dynamic Channel Assignment in CBRS Band}

\author{Xuhang~Ying,~\IEEEmembership{Student Member,~IEEE,}
	Milind~M.~Buddhikot,~\IEEEmembership{Member,~IEEE,}
        and~Sumit~Roy,~\IEEEmembership{Fellow,~IEEE,}


\thanks{Xuhang Ying and Sumit Roy are with the Department of Electrical Engineering, University of Washington, Seattle, WA, USA. Email: \{xhying, sroy\}@uw.edu. Milind M. Buddhikot is with the Nokia Bell Labs, Murray Hill, NJ, USA. Email: milind.buddhikot@nokia-bell-labs.com.
Part of this work was present in IEEE DySPAN 2017 \cite{ying2017coexistence}. 
}

}

\maketitle

\begin{abstract}
The paradigm of shared spectrum allows secondary devices to opportunistically access spectrum bands underutilized by primary owners. 
\rev{Recently}, the FCC \rev{has} targeted \rev{the} sharing \rev{of} the 3.5 GHz (3550-3700 MHz) federal spectrum with commercial systems such as small cells. 
The 
rules require a Spectrum Access System (SAS) to \rev{accommodate three service tiers: 1) Incumbent Access,} 2) Priority Access (PA), and 3) Generalized Authorized Access (GAA).
In this work, we study the SAS-assisted dynamic channel assignment (CA) for \rev{PA and GAA tiers}. 
We introduce the node-channel-pair conflict graph to capture pairwise interference, channel and geographic contiguity constraints, spatially varying channel availability, and coexistence \rev{awareness}. 
\rev{The proposed graph representation allows us to} formulate PA CA and GAA CA with binary conflicts as max-cardinality and max-reward CA, respectively.
\rev{Approximate solutions can be found by a heuristic-based algorithm that searches for the maximum weighted independent set.}
\rev{We further formulate GAA CA with non-binary conflicts  as max-utility CA.}
\rev{We show} that the utility function is submodular, and the problem is an instance of matroid-constrained submodular maximization. 
A local-search-based polynomial-time algorithm is proposed that provides a provable performance guarantee.
\rev{Extensive simulations using a real-world Wi-Fi hotspot location dataset are performed} to evaluate the proposed algorithms.
Our results have demonstrated the advantages of the proposed graph representation and improved performance of the proposed algorithms over the baseline algorithms.

\end{abstract}


\markboth{IEEE TRANSACTIONS ON WIRELESS COMMUNICATIONS, VOL. X, NO. X, XX 201X}
{Working paper}


\section{Introduction}
The rapidly increasing demand in wireless capacity with future ultra-broadband wireless networks 
is driving regulatory bodies to pursue policy innovations based on the paradigm of \textit{shared spectrum}. 
Such paradigm allows secondary devices to access licensed spectrum bands that are underutilized by primary owners.
In \cite{fcc2008TVWS}, the FCC allowed the unlicensed use of underutilized TV spectrum as a first step. 
In the 2012 PCAST report \cite{PCAST2012}, the U.S. government pushed it further by announcing its intentions to explore spectrum sharing policies and innovative technologies to share 1000 MHz federal government spectrum with commercial systems. 
Towards this goal, the FCC targeted \rev{the} release of 3550-3700 MHz band for small cell deployment, termed \rev{the} \textit{Citizens Broadband Radio Service} (CBRS), which is primarily used by high-power Department of Defense shipborne radars and non-federal Fixed Satellite Services (FSS) earth stations.

\begin{figure}[t!]
    \centering
    \includegraphics[width=1\columnwidth]{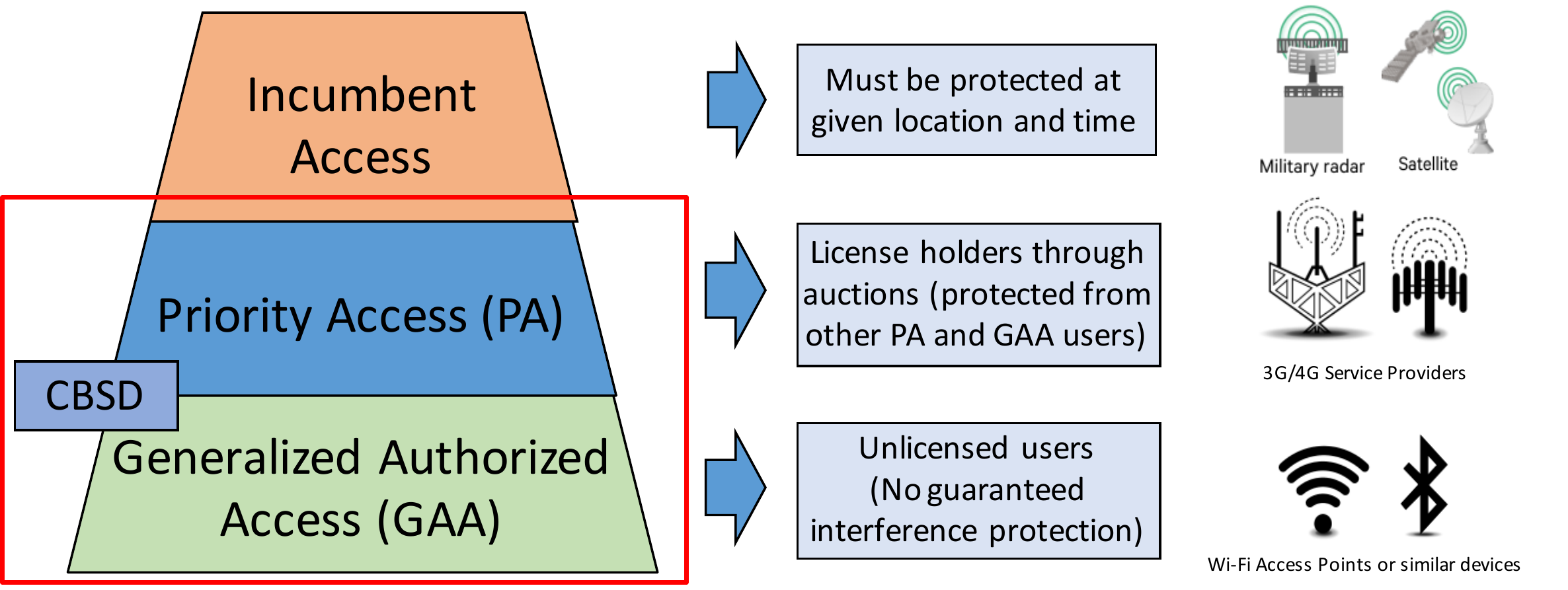}
    \caption{Three-tiered spectrum access in 3.5 GHz CBRS band. }
    \label{fig:three_tier_framework}
\end{figure}

In \rev{the} recent ruling \cite{fcc2012,fcc2014,fcc2015,winnforum2018requirements} for \rev{the} CBRS, the FCC proposed a spectrum \rev{access} framework that consists of three tiers: 1) \textit{Incumbent Access}, 2) \textit{Priority Access} (PA), and 3) \textit{Generalized Authorized Access} (GAA), as illustrated in Fig.~\ref{fig:three_tier_framework}. 
PA and GAA devices are also referred to as \res{Citizens Broadband Radio Service Devices (CBSDs), which are fixed stations or networks of such stations operating on a PA or GAA basis. End user devices are not considered as CBSDs.}
As we can see, the above framework highlights the hierarchical spectrum access rights: 
Incumbents have the highest priority and 
would be protected from harmful interference from all other users by forbidding 
CBSD transmissions within and close to the \rev{activated \textit{Dynamic Protection Areas} (DPAs)}.
\rev{PA users with critical quality-of-service needs (e.g., hospitals and public safety entities) would be authorized to operate at specific locations with certain interference protection.} \rev{ 
But they need to purchase} PA licenses (PALs) \rev{via competitive bidding}, each of which guarantees the authorized and protected use of a $10$ MHz channel in a census tract.
In contrast, GAA CBSDs (or GAA nodes) would be authorized to 
access the shared bands within designated geographic areas, but they \rev{would expect no interference protection and should avoid causing interference to incumbents and PA users}. 
In the center of the CBRS ecosystem lies the \textit{Spectrum Access System} (SAS) \rev{that authorizes and manages use of spectrum for the CBRS.}
\rev{In this work, we study the SAS-assisted dynamic channel assignment (CA) for PA and GAA tiers.}

The 3.5 GHz band is divided into $15$ orthogonal $10$ MHz channels. 
Up to $7$ channels can be assigned to PALs in a \textit{license area} (a census tract), 
and a licensee is allowed to aggregate up to $4$ channels by stacking multiple \revv{PALs} in a \textit{service area} (one or more contiguous license areas).
\rev{While each PAL guarantees one channel, the exact channel assignment is not fixed and will be determined by the SAS.}
Since the FCC requires the SAS to assign \textit{contiguous channels to geographically contiguous PALs}, it imposes a major challenge to the design of a PA CA scheme.

To improve spectrum utilization, a PAL channel is made available for GAA use at locations outside the \textit{PAL protection areas} (PPAs) and their vicinities under the ``use-it-or-share-it'' rule.
\rev{As a result}, \textit{channel availability} becomes location dependent and may vary significantly among GAA nodes.
Moreover, a GAA node may also request multiple \textit{contiguous channels} to meet its capacity demand or support network operations. 
For instance, a Wi-Fi-based node would require at least two contiguous channels in the CBRS. 

In this work, we are interested in co-channel coexistence enabled by Wi-Fi like MAC protocols, such as CSMA/CA \cite{80211ac2013} and Listen-Before-Talk in LTE-LAA \cite{etsi2016lte}. 
\rev{Hence,} the pairwise interference relationship of GAA nodes can be classified into three cases: 1) no conflict, 2) type-I conflict -- two interfering nodes are hidden from each other and cannot detect \res{possible interference at end user devices (or clients) without their feedbacks (e.g., packet loss or throughput degradation)}, \rev{or} 3) type-II conflict -- two interfering nodes are within each other's carrier-sensing (CS) or energy-detection (ED) range and can resolve the conflict through contention. 
Since GAA nodes are able to harmoniously share the same channel under type-II conflicts, 
the SAS \rev{can} exploit such \textit{coexistence opportunities} to accommodate more devices \rev{in a dense network}.

Despite that dynamic CA has been studied in the context of spectrum sharing from various perspectives such as graph coloring \cite{mishra2005weighted, subramanian2007fast, subramanian2008near, hessar2014resource} and
game theory \cite{nie2006adaptive,liu2008spectum,
li2012distributed},
the aforementioned challenges (i.e., channel and geographic contiguity, spatially varying channel availability, and coexistence \rev{awareness}) differentiate dynamic CA in the CBRS from previous work. 
In this paper, we make the following specific  contributions:
\begin{itemize}
\item We define the \textit{node-channel pair} (NC pair) that assigns a set of contiguous and available channels to a node (i.e., a PA service area or a GAA CBSD) and introduce the \textit{NC-pair conflict graph}, in which each vertex is a NC pair and each edge represents a conflict. 

\item With the proposed conflict graph, we formulate PA CA as \textit{max-cardinality CA}. 
For GAA CA with binary conflicts, we extend NC pairs to \textit{super-NC pairs} to exploit type-II conflicts, where each super-node consists of a set of GAA nodes that can detect each other's transmission.
We define a reward function to capture the preferences of (super-)NC pairs and formulate binary GAA CA as \textit{max-reward CA}. 
To further enhance coexistence awareness, 
we define a penalty function and extend the binary conflicts to non-binary.
The non-binary GAA CA is then formulated as \textit{max-utility CA} to trade off rewards against penalties.

\item We propose a super-node formation algorithm based on clique searching and bin packing. 
The max-cardinality and max-reward CA problems are mapped to the problem of finding the maximum (weighted) independent set, and approximate solutions are obtained through a heuristic-based algorithm.
For max-utility CA, we show that the utility function is submodular, and our problem is an instance of matroid-constrained submodular maximization.
We propose a polynomial-time algorithm based on (approximate) local search that provides a provable performance guarantee. 

\item We conduct extensive simulations \rev{using a real-world dataset that contains Wi-Fi hotspot locations in the New York City} \cite{NYCWiFiLocation} to evaluate the proposed algorithms.
For PA CA, our results show that the proposed algorithm consistently serves over $93.0\%$ of service areas and outperforms the baseline algorithm by over $30.0\%$.
For binary GAA CA, the proposed algorithm with linear rewards \rev{is able to accommodate $10.2\%$ more nodes and serve $10.4\%$ more demand on average than the baseline algorithm}. \rev{Besides, enabling coexistence awareness can effectively improve the performance of the proposed algorithms}. 
For non-binary GAA CA, \rev{the proposed algorithm achieves $29.5\%$ more utility while keeping the total interference much smaller than \revv{random selection}}.
\revv{We also show that an operator can leverage the proposed framework to optimize the overall capacity of GAA networks.}


\end{itemize}

The remainder of this paper is organized as follows.
Section~\ref{sec:related_work} reviews the related work on CA.
Section~\ref{sec:channel_assignment_scenario} summarizes CA scenarios and presents the SAS-based architecture.
Conflict graphs and problem formulations are presented in Section~\ref{sec:conflict_graph}, and proposed algorithms are provided in Section~\ref{sec:algo}.
Section~\ref{sec:evaluation} presents simulation results and Section~\ref{sec:conclusion} concludes this study.

\section{Related Work}\label{sec:related_work}
Dynamic CA has been studied in different contexts including cellular networks \cite{katzela1996channel}, mesh networks \cite{si2010overview} and cognitive radio networks \cite{akyildiz2006next}.
Among the proposed approaches, \rev{graphs are widely adopted for CA modeling.}
In \cite{cao2005distributed, wang2005list, peng2006utilization,hessar2014resource}, authors considered CA \rev{under the} channel availability and interference constraints \rev{using} vertex-weighted interference/conflict graphs, where each vertex represents a user with a color \rev{(channel)} list, and each edge denotes a binary conflict due to co-channel interference.
The problem is then formulated as graph multi-coloring (i.e., assign 
conflict-free colors to users and obtain rewards on vertices) and a typical objective is to maximize the total reward or utility (derived from rewards).
Such conflict-free CA algorithms 
are suitable for sparse networks with QoS consideration for each user.

In order to capture the interference with greater granularity, non-binary conflicts using an edge-weighting or penalty function may be introduced, and typical objectives include minimizing the total interference or penalty via graph coloring \cite{mishra2005weighted, subramanian2007fast, yue2011cacao} or maximizing the total reward subject to the aggregate interference limit at each user \cite{hoefer2014approximation,hessar2014resource}. 
However, \rev{the vast majority} of existing works do \textit{not} consider geographic and channel contiguity constraints \rev{and thus are} not applicable to CA in the CBRS. 
In \cite{subramanian2008near}, Subramanian \textit{et. al.} introduced the notion of channel graph to enforce channel contiguity,  
where each vertex is a channel consisting of several contiguous primitive channels and each edge \rev{indicates overlapping channels.}
\rev{In Section~\ref{sec:evaluation_max_reward_CA}, we will show that our proposed algorithm outperforms the algorithm in \cite{subramanian2008near} in terms of accommodating more users and meeting greater demand. 
}


\begin{figure*}[!ht]
	\centering
	\subfloat[]{\includegraphics[width=1\columnwidth]{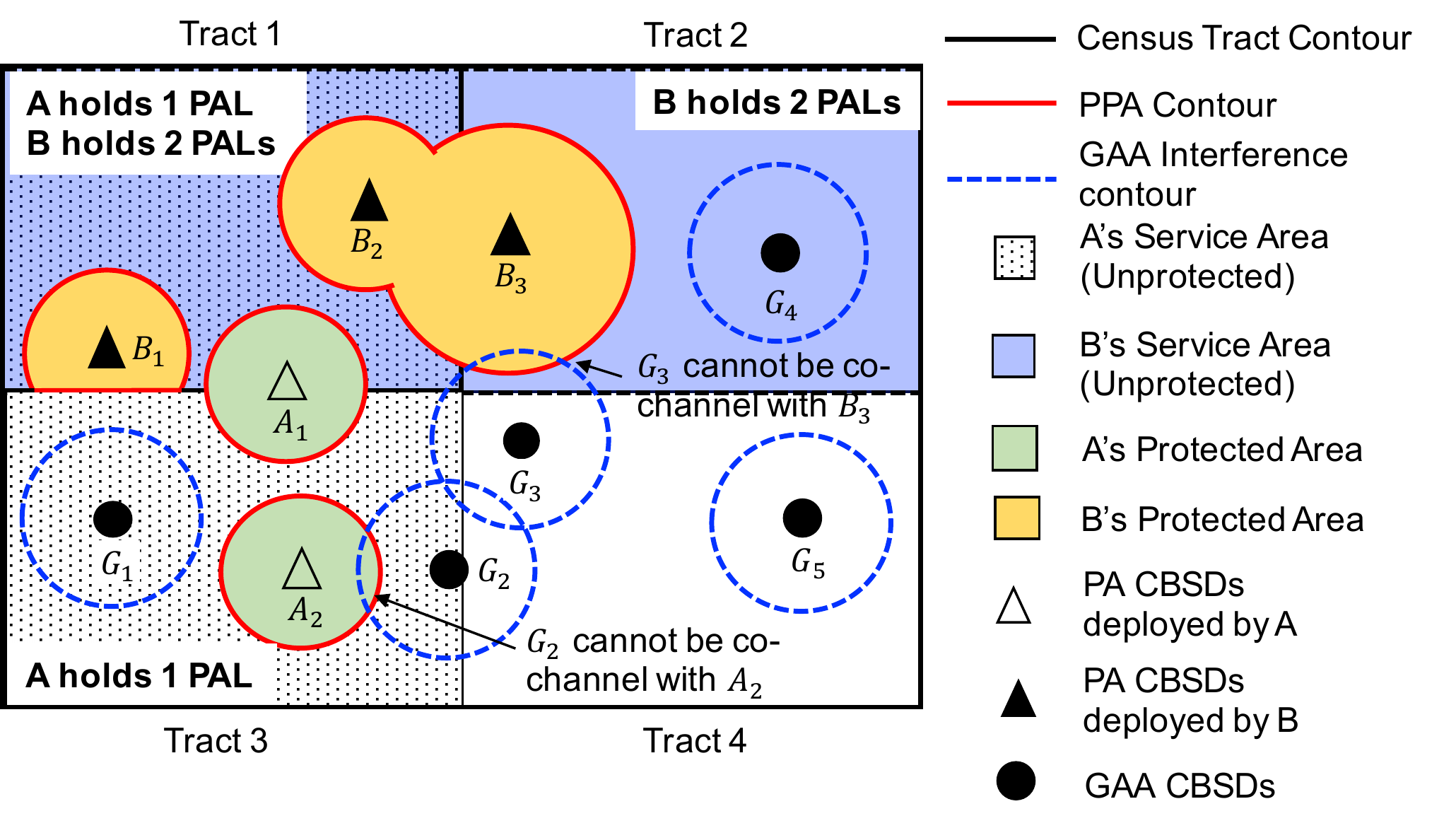}\label{fig:channel_assignment_example1}}
	\subfloat[]{\includegraphics[width=.8\columnwidth]{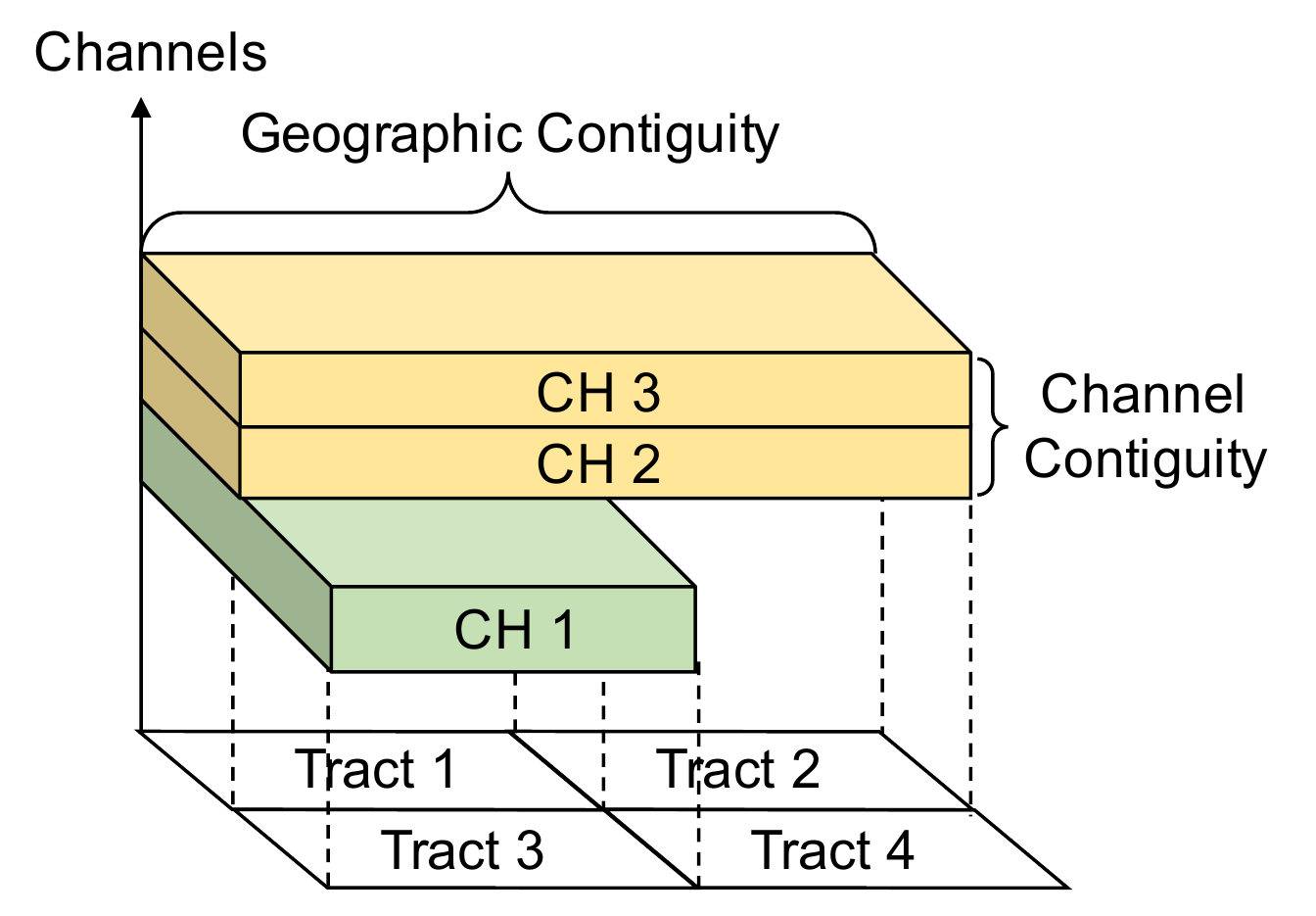}\label{fig:channel_assignment_example2}}
	\caption{Illustration of the PA CA scenario. 
		(a) A holds one PAL in its service area that consists of tracts 1 and 3, while B stacks two PALs in its service area that consists of tracts 1 and 2.  
		Each PA CBSD deployed by the licensee has a PPA that cannot extend beyond its service area. 
		GAA nodes may reuse assigned PAL channels at locations beyond the PPAs, subject to the interference constraint. 
		(b) Given three PAL channels, the SAS can assign CH $\{1\}$ in tracts $1$ and $3$ to A's PALs and CH $\{2,3\}$ in tracts $1$ and $2$ to B's PALs, so as to meet the geographic and channel contiguity requirements.
	}
	\label{fig:channel_assignment_example}
\end{figure*}

\rev{In this work, we incorporate the contiguity constraints by introducing (super)NC-pair conflict graphs, which can also be used to represent channel availability and enhance coexistence awareness. 
Besides, the proposed conflict graph allows us to relate the PA CA and binary GAA CA to the classic the maximum (weighted) independent set problem \cite{halldorsson2000approximations,sakai2003note,kako2005approximation} and adopt a widely used heuristic-based algorithm that runs in polynomial time with a performance guarantee.
To further enhance coexistence awareness, we consider non-binary GAA CA and formulate the problem as utility maximization.}
We then prove the submodularity of the utility function and
propose a polynomial-time algorithm based on local search with a provable performance guarantee.


%
%

\section{Channel Assignment \rev{in CBRS}}\label{sec:channel_assignment_scenario}
In this section, we describe CA scenarios for PA and GAA tiers based on the FCC rules \cite{fcc2012, fcc2014, fcc2015, fcc2016}.
In particular, we highlight CA challenges and opportunities in the CBRS: geographic and channel contiguity, spatially varying channel availability, flexible demands, and coexistence awareness. 
\rev{Then we present the architecture of SAS from the perspective of dynamic CA. }

\subsection{PA CA Scenario}\label{sec:PA_channel_assignment}
As per the FCC rules, 
each deployed PA CBSD has an associated PPA with a default protection contour that is calculated based on the signal level of $-96$ dBm/$10$ MHz, which cannot extend beyond the licensee's service area.
A self-reported PPA contour is also acceptable, so long as it is within the default PPA contour.
If PPAs for multiple CBSDs operated by the same licensee overlap, they would be merged into a single PPA. 
Note that interference protection is \rev{enforced} for each active PPA, that is, the aggregate co-channel interference from other PA or GAA CBSDs at any location within a PPA cannot exceed $-80$ dBm/$10$ MHz. 



The FCC imposes two contiguity requirements for PA CA: 1) \textit{geographic contiguity} -- an SAS must assign geographically contiguous PALs held by the same PA licensee to the same channels in each geographic area, to the extent feasible; and 2) \textit{channel contiguity} -- an SAS must assign multiple channels held by the same PA licensee to contiguous channels in the same license area, to the extent feasible.
An example is provided in Fig.~\ref{fig:channel_assignment_example} that illustrates the PA CA scenario. 

%



\subsection{GAA CA Scenario}

GAA CA \rev{differs} from PA CA in the following aspects. 
First, channel availability for GAA nodes can vary significantly depending on location.
This is because a PAL channel is considered ``in use'' only within PPAs and becomes available for GAA use at a location beyond the PPAs (and their vicinities), 
if GAA transmission at that location does not cause 
\rev{over-the-limit interference} 
at any location within PPAs. 
As illustrated in the example in Fig.~\ref{fig:channel_assignment_example}, GAA nodes $G_1$ and $G_4$ are located in the service area of PA licensees, but they are free to use all channels just like $G_5$, since they do not cause significant interference to any PA CBSD. 
In contrast, $G_2$ cannot operate co-channel with $A_2$, and $G_3$ cannot operate co-channel with $B_3$, due to the interference constraint. 
Although it is not desirable for $G_2$ and $G_3$ to operate co-channel as they would interfere with each other, it is indeed acceptable as per the FCC ruling.


Second, a GAA node may request multiple contiguous channels (up to a certain limit such as four), \rev{and the demand could be very flexible}. 
For instance, a Wi-Fi node may request two contiguous channels but would be very willing to receive four to obtain higher capacity.
Although the SAS is not obligated to meet the maximum demands of all GAA nodes, it \rev{would be expected to} maximize the number of assigned channels or the overall throughput with best efforts.

\begin{figure*}[h!]
	\centering
	\includegraphics[width=.7\textwidth]{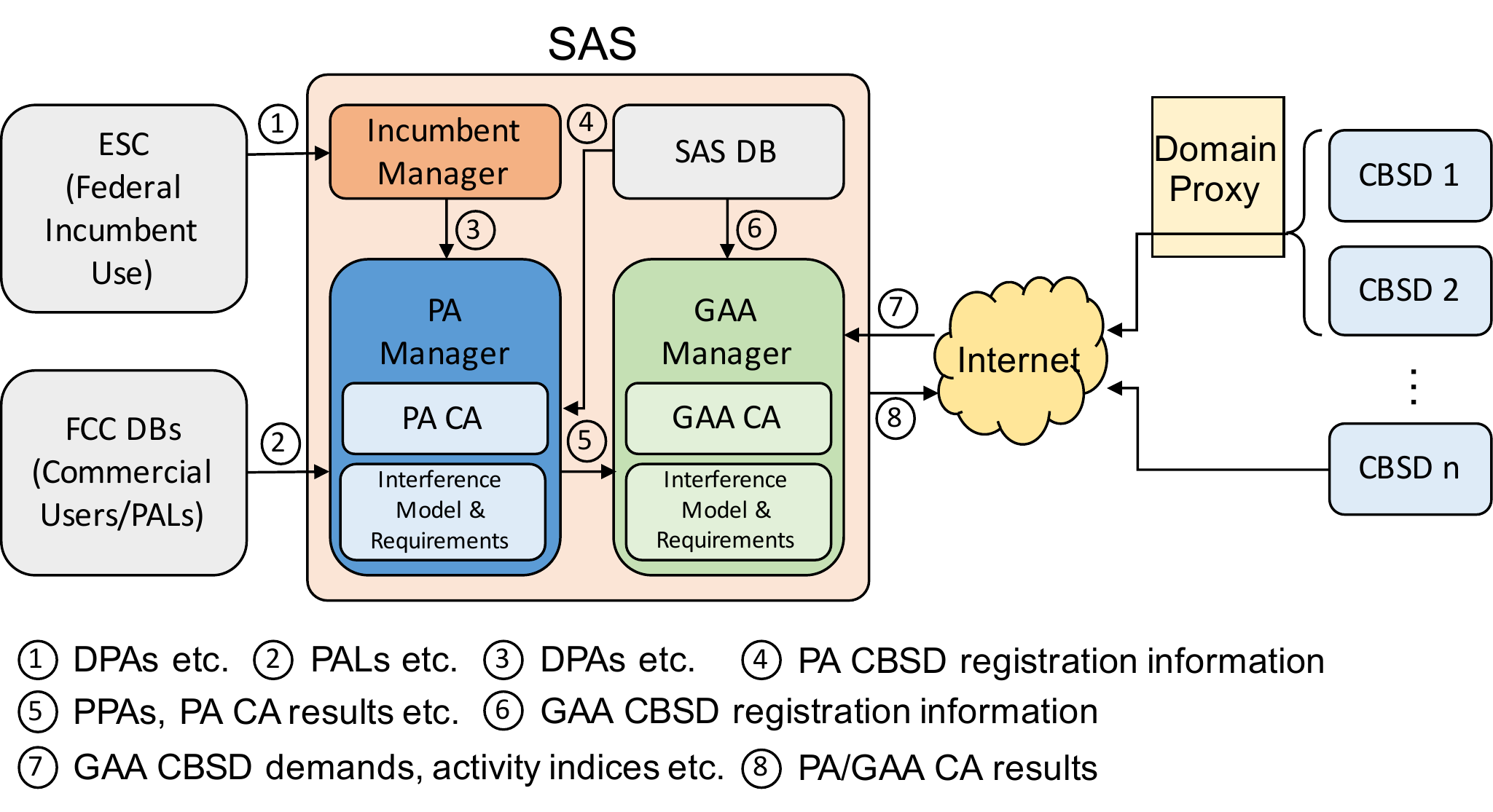}
	\caption{Architecture of SAS as a centralized entity from the perspective of dynamic CA. A centralized PA/GAA CA algorithm will be implemented in the PA/GAA manager within the SAS. Note that only data flows for PA/GAA CA purposes are highlighted.}
	\label{fig:SAS_assisted_CA_architecture}
\end{figure*}

Third, GAA nodes are expected to coexist in the same frequency, space and time. 
In this work, we consider co-channel coexistence enabled by Wi-Fi like MAC protocols such as CSMA/CA in Wi-Fi \cite{80211ac2013} and Listen-Before-Talk in LAA \cite{etsi2016lte}. 
Traditional channel assignment techniques \cite{cao2005distributed, wang2005list, peng2006utilization} typically model interference as a binary conflict and avoid assigning the same channel to any two conflicting nodes.
However, two interfering nodes in close proximity are able to detect each other's transmission (i.e., no hidden nodes) and share the same channel(s) in the CSMA/CA fashion.  
Hence, the SAS may exploit such opportunities to accommodate more nodes in this case, especially when the spectrum is crowded. 
In this work, we further consider non-binary conflicts in terms of penalties to enhance \textit{coexistence awareness}. 

%


\res{
\subsection{Architecture of SAS-Assisted Dynamic CA}
As shown in Fig.~\ref{fig:SAS_assisted_CA_architecture}, the SAS is a centralized entity consisting of incumbent, PA and GAA managers.
It communicates with the Environment Sensing Capability (ESC) component, FCC databases (DBs), and CBSDs. 
The PA manager receives information on DPAs and PALs, determines PAL channel availability based on the interference model and requirements, and executes the PA CA algorithm to assign channels to PALs grouped by service areas.  
The GAA manager takes as inputs 1) the PPAs and PA CA results, 2) GAA CBSD registration information (e.g., location, transmit power, and antenna gain), 3) CBSD demands and activity indices (i.e., traffic load indicators), and 4) the interference model\footnote{Radio propagation models, despite their limitations, are still widely used for initial network planning, but the interference relationship can also be determined from real measurements. Such measurement reporting is already supported by the WInnForum SAS-CBSD protocol \cite{winnforum2018signaling}.} and requirements. 
It then determines channel availability for GAA nodes and executes the GAA CA algorithm. 
Finally, the SAS will disseminate the PA and GAA CA results to individual CBSDs (e.g., within 10s of seconds). 

According to the WInnForum SAS-CBSD protocol specifications \cite{winnforum2018signaling}, the protocol exchange between the SAS and CBSDs already exists, including requests and responses for (de)registration, spectrum inquiry, grant (relinquishment) and heatbeat.
Therefore, the SAS can readily leverage registration and regular heartbeat messages to collect inputs from CBSDs and use them for CA purposes, without introducing additional signaling overhead. 

It is also important to note that the CA in the CBRS is indeed dynamic:  changes in DPAs and PALs would trigger PA CA, which further trigger GAA CA along with other factors (e.g., changes in GAA CBSD demands). 
Since DPA and PAL changes are relatively infrequent, PA channel assignments will be expected to be quasi-stable, and channel reassignment occurs presumably in the order of hours or days. In contrast, the SAS may have to aggregate demands from all GAA nodes and perform GAA CA periodically (e.g., every 100s of seconds).
Nevertheless, the SAS can always take a snapshot of the CBSD network at a given time instant and perform CA for PA and GAA tiers separately in a centralized fashion\footnote{The SAS itself may be implemented in a distributed cloud so as to exploit dedicated computing resources spread over multiple virtual machines to perform CA for a specific region.}. }

\section{Novel Conflict Graphs and Problem Formulation}\label{sec:conflict_graph}

\rev{In order to incorporate channel availability and contiguity}, we define a NC pair as a channel assignment that assigns a set of contiguous and available channels to a node, i.e., a PA service area or a GAA node.
We then \rev{introduce} the NC-pair conflict graph, where each vertex is a NC pair and each edge indicates a conflict. 
In \rev{the rest of this section}, we will describe conflict graph construction in detail and present our problem formulation for each CA scenario.

\subsection{Conflict Graph for PA CA}
\rev{Let us} consider a geographic region with $M$ PA licensees. 
Each licensee $i$ holds $N_i$ PALs in $n_i$ service areas,
and there are a total of $n=\sum_{i=1}^M n_i$ service areas in total. 
Let $S_{ij}$ be the set of license areas that form the $j$-th service area of licensee $i$, and each license area in $S_{ij}$ has $N_{ij}$ PALs\footnote{The constraint that no more than seven PALs shall be assigned in any license area should have been enforced during the licensing stage and can be verified by checking $\sum_{i=1}^M \sum_{j=1}^{n_i} \bold{1}_{\{s \in S_{ij}\}} \cdot N_{ij} \leq 7$ for each license area $s$, where $\bold{1}_{\{s \in S_{ij}\}}$ is the indicator variable.}.
There are a total of $10$ PAL channels in the CBRS, denoted as $\Omega=\{1,2,...,10\}$.
Due to the \rev{DPAs} of incumbents, the set of available PAL channels for $S_{ij}$, denoted as $\Gamma(S_{ij})$, is a subset of $\Omega$. 
Denote the set of valid channel assignments for $S_{ij}$ as $\mathcal{C}_{ij}$. 
It means that each channel assignment $C \in \mathcal{C}_{ij}$ consists of contiguous channels and satisfies the conditions $C\in \Gamma({S_{ij}})$ and $|C| = N_{ij}$. 
For instance, with $\Gamma(S_{ij})=\{1,2,3\}$ and $N_{ij}=2$, we have $\mathcal{C}_{ij}=\{ \{1,2\}, \{2,3\}\}$.

In the proposed conflict graph $G=(V,E)$ for PA CA, each vertex is a NC pair $(S_{ij}, C_{ij})$ that assigns channel(s) $C_{ij}$ to $S_{ij}$, and an edge exists between $(S_{ij}, C_{ij})$ and $(S_{lk}, C_{lk})$, if
\begin{itemize}
\item $i=j$ and $l=k$ (i.e., $S_{ij}$ and $S_{lk}$ refer to the same service area), and $C_{ij} \neq C_{lk}$, or


\item $S_{ij}\cap S_{lk} \neq \emptyset$ and $C_{ij} \cap C_{lk} \neq \emptyset$.

\end{itemize} 
The first constraint requires each service area to take at most one channel assignment, referred to as the \textit{one-channel-assignment-per-node constraint}, while
the second constraint prevents two overlapping service areas from being assigned overlapping channels, which is called the \textit{conflict constraint}. 
\res{It is worth noting that by selecting a vertex $(S_{ij}, C_{ij})$, that is, assigning channel(s) $C_{ij}$ to $S_{ij}$, the SAS can meet both the geographic and channel contiguity requirements.}
An example of the proposed conflict graph for PA CA is provided in Fig.~\ref{fig:PA_conflict_graph_ex}.

\begin{figure}[ht!]
	\centering
	\includegraphics[width=0.6\columnwidth]{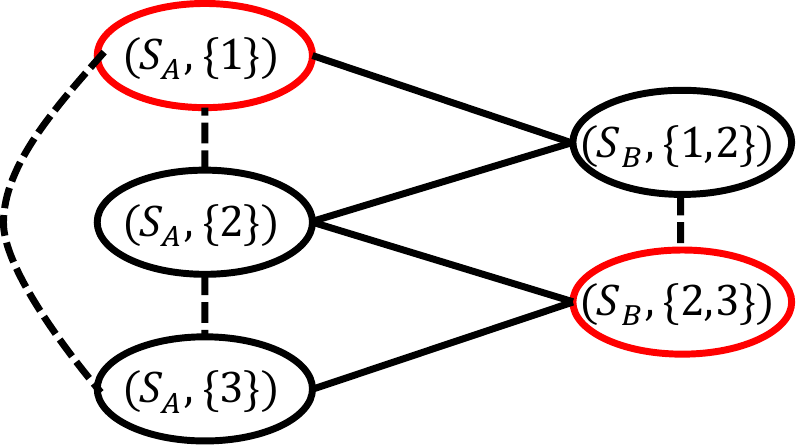}
	\caption{
		Conflict graph for the PA CA example in Fig.~\ref{fig:channel_assignment_example}. 
		Dashed and solid edges are due to the one-channel-assignment-per-node and conflict constraints, respectively. 
		The SAS aims to choose the largest set of conflict-free vertices. 
		One possible CA scheme consists of two vertices in red, i.e, $(S_A, \{1\})$ and $(S_B, \{2,3\})$. 
		An alternative is $(S_A, \{3\})$ and $(S_B, \{1,2\})$.
	}
	\label{fig:PA_conflict_graph_ex}
\end{figure}

\res{Since a conflict due to interference exists only for two NC pairs that represent two overlapping service areas with overlapping channels, a sparse conflict graph would be expected in practice.
Therefore, adjacency lists would be preferred in general with a storage cost of $O(|V|+|E|)$. It takes $O(1)$ time to add a vertex or edge, $O(E)$ to remove a vertex, and $O(V)$ to remove an edge.
More details are available in \cite{cormen2009introduction}.}

%

\noindent \textbf{Problem Formulation:}
Given a conflict graph $G(V,E)$ for PA CA, the SAS wants to find a scheme $I \subseteq V$ so as to maximize the total number of served service areas without conflicts,
\begin{align}
\max_{I \subseteq V} &~|I| \nonumber \\
\text{Subject to: } e(u, v) & \notin E, \forall u,v\in I,\label{eq:max_cardinality_CA}
\end{align}
where $|\cdot|$ is the cardinality operator.
We refer to the above problem as \textbf{Max-Cardinality CA}. 

\subsection{Conflict Graph with Binary Conflicts for GAA CA }
\label{sec:conflict_graph_binary}
Suppose that there are $n$ GAA nodes at fixed locations in the same geographic region. 
The ground set of channels is $\Omega = \{1,2,...,15\}$, and the SAS can determine the set of available channels $\Gamma(i) \subseteq \Omega$ 
for each GAA node $i$. 
In this paper, we focus on CA by assuming fixed transmit power for all nodes and leave joint channel and power assignment as future work.

Denote \rev{the demand set of} node $i$ as $\mathcal{D}(i)=\{d_{i1}, d_{i2},...\}$, where $d_{ik}$ is the number of contiguous channels \rev{node $i$ requests}. 
For example, $\mathcal{D}(i)=\{2,4\}$ means that node $i$ requests for two or four contiguous channels.
By default, the SAS may set $\mathcal{D}(i)$ to $\{1,2,...,L\}$, where $L$ is the limit (e.g., $L=4$). 
As a result, $\Gamma(i)$ and $\mathcal{D}(i)$ jointly determine the set of valid channel assignments $\mathcal{C}(i)$.
For example, with $\Gamma(i)=\{1,4,5\}$ and $\mathcal{D}(i)=\{1,2\}$, we have $\mathcal{C}(i) = \{\{1\}, \{4\}, \{5\}, \{4,5\}\}$. 


In the proposed conflict graph $G(V,E)$ for \rev{binary} GAA CA, each vertex is a NC pair $(i, C_i)$ that assigns channel(s) $C_i \in \mathcal{C}(i)$ to node $i$, and an edge exists between $(i,C_i)$ and $(j,C_j)$ if 
\begin{itemize}
\item $i=j$ and $C_i \neq C_j$, or 

\item $i\neq j$, $C_i \cap C_j \neq \emptyset$, and node $i$ would conflict or interfere with node $j$ when assigned the same channel based on the interference model.
\end{itemize}
\rev{The above two conditions correspond to the one-channel-assignment-per-node constraint and the conflict/interference constraint, respectively.}
\res{Like the conflict graph for PA CA, the conflict graph for GAA CA is expected to be sparse and thus adjacency lists are generally preferred.}

In the rest of this section, we will characterize pairwise interference relationship, discuss coexistence awareness under binary conflicts, and introduce a reward function to reflect preferences of NC pairs. 
Then we will present our problem formulation for binary GAA CA.




\begin{figure*}[!t]
	\centering
	\subfloat[]{\includegraphics[width=.3\textwidth]{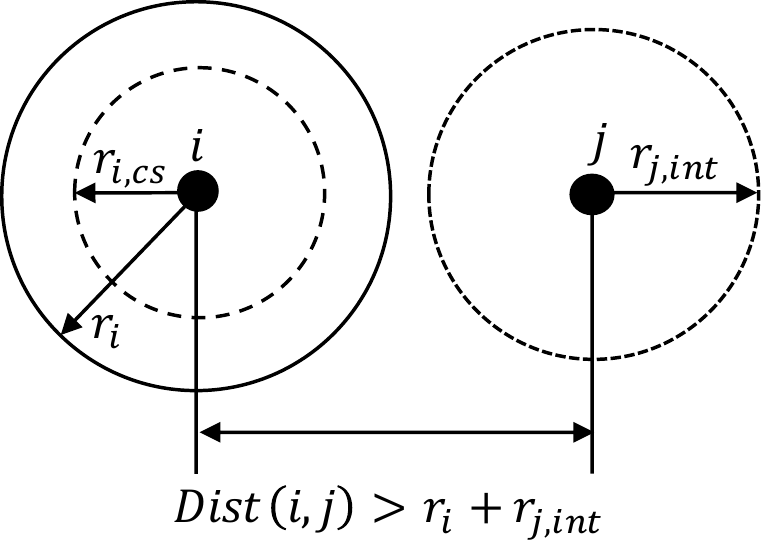}
		\label{fig:protocol_model_case1}}
	~
	\subfloat[]{\includegraphics[width=.24\textwidth]{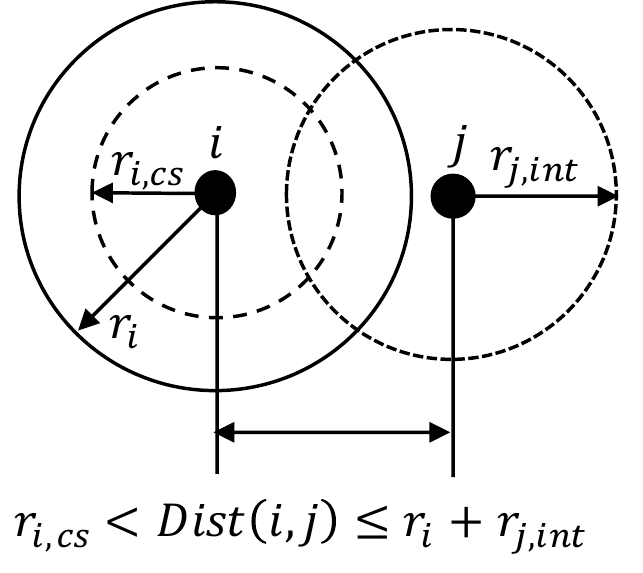}
		\label{fig:protocol_model_case2}}
	~~~~~~
	\subfloat[]{\includegraphics[width=.18\textwidth]{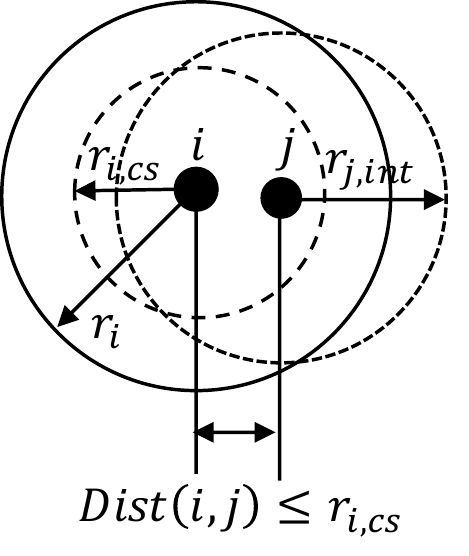}
		\label{fig:protocol_model_case3}}
	\caption{Impact of node $j$'s interference on node $i$. 
		(a) No conflict. Nodes $i$ and $j$ are free to reuse the same channel if $j$'s interference region does not overlap with $i$'s service region. 
		(b) Type-I conflict. Node $j$ causes non-negligible interference to $i$'s clients located in the overlapping region, which cannot be detected by $i$ \rev{without feedback from the clients such as packet loss or throughput degradation}. 
		(c) Type-II conflict. Node $j$ is within $i$'s CS or ED range, and its interference can detected by $i$.
	}
	\label{fig:protocol_model}
	
\end{figure*}

\subsubsection{Characterization of Pairwise Interference}\label{sec:interference_model}
First of all, the SAS needs to determine whether two GAA nodes are interfering with each other, which is a binary (yes-or-no) decision. 
Similar to the interference protection for PA users, the SAS may decide that node $j$ is interfering with $i$, if the received interference from node $j$ at any location within $i$'s service contour (e.g., the $-96$ dBm/$10$ MHz contour) exceeds a certain limit (e.g., $-80$ dBm/$10$ dBm).
Denote node $i$'s service region as $\mathcal{R}_i$ of radius $r_i$ and node $j$'s interference region as $\mathcal{R}_{j,int}$ of radius $r_{j, int}$.
Then there exists no conflict between $i$ and $j$ \rev{due to $j$' interference}, if the distance is larger than $r_i + r_{j,int}$, regardless of the positions of $i$'s clients (Fig.~\ref{fig:protocol_model_case1}).

As nodes $i$ and $j$ get closer,  $\mathcal{R}_{j,int}$ starts to overlap with $\mathcal{R}_{i}$, which means that $j$ may cause interference to clients that are possibly located in the overlapping region in the downlink. 
If the distance is smaller than $r_i + r_{j,int}$ but larger than $i$'s CS/ED range $r_{i,cs}$, as shown in Fig.~\ref{fig:protocol_model_case2}, node $j$ is said to be hidden from $i$, and the interference from $j$ cannot be detected by $i$ \rev{without feedback from $i$'s clients}.
When the distance is less than $r_{i,cs}$ (Fig.~\ref{fig:protocol_model_case3}), node $i$ is able to detect $j$'s transmission via CS/ED and 
achieve 
co-channel coexistence \rev{through contention}. 
In order to distinguish the above cases, we call them type-I and type-II conflicts, respectively. 
\rev{Note that} our interference model is very similar to the popular protocol models in literature (e.g., \cite{gupta2000capacity,mishra2006client, niculescu2007interference}), \rev{which} can be easily implemented in the SAS using radio propagation models.

\subsubsection{Coexistence Awareness}\label{sec:coexistence_awareness}
Ideally, two nodes should be assigned the same channel only when they do not conflict.
Nevertheless, \rev{GAA nodes have the lowest spectrum access priority} and are expected to coexist in the presence of conflicts.
Compared with type-I conflicts, type-II conflicts can be  handled more easily and thus may be considered as coexistence opportunities to the SAS.  
To that end, we introduce \textit{super-NC pair} $(S,C)$, where $S$ is \rev{a \textit{super-node}} that consists of nodes that can coexist on channel $C$ \rev{with two conditions: 1) channel $C$ is available at all nodes in $S$ and 2)} nodes in $S$ be within each other's CS/ED range so that all conflicts can be gracefully resolved.
The super-node formation algorithm will be presented later in Section~\ref{sec:super_node_formation}.

Fig.~\ref{fig:GAA_conflict_graph_ex} shows an example of the proposed conflict graph \rev{with and without coexistence awareness}.
\rev{Suppose that} nodes B, C are within each other's CS/ED range. \rev{Then} there are two super-NC pairs $(\{B,C\}, \{1\})$ and $(\{B,C\}, \{2\})$.
\rev{As illustrated in Fig.~\ref{fig:GAA_interference_graph_with_coex}}, 
when a super-NC pair is identified, it is added to the graph and inherits the conflict relationship of its children RC pairs. 
Then edges among its children RC pairs are removed.  
For instance, removing the edge between $(B,\{2\})$ and $(C,\{2\})$ means that \rev{when} $(B,\{2\})$ is selected, $(\{B,C\},\{2\})$ \rev{will become} invalid, but the SAS can still select $(C,\{2\})$, \rev{since B and C are able to resolve type-II conflicts. } 

\begin{figure}[!t]
	\centering
	\subfloat[]{\includegraphics[trim=0 0.8cm 0 0.8cm,clip=true,width=.7\columnwidth]{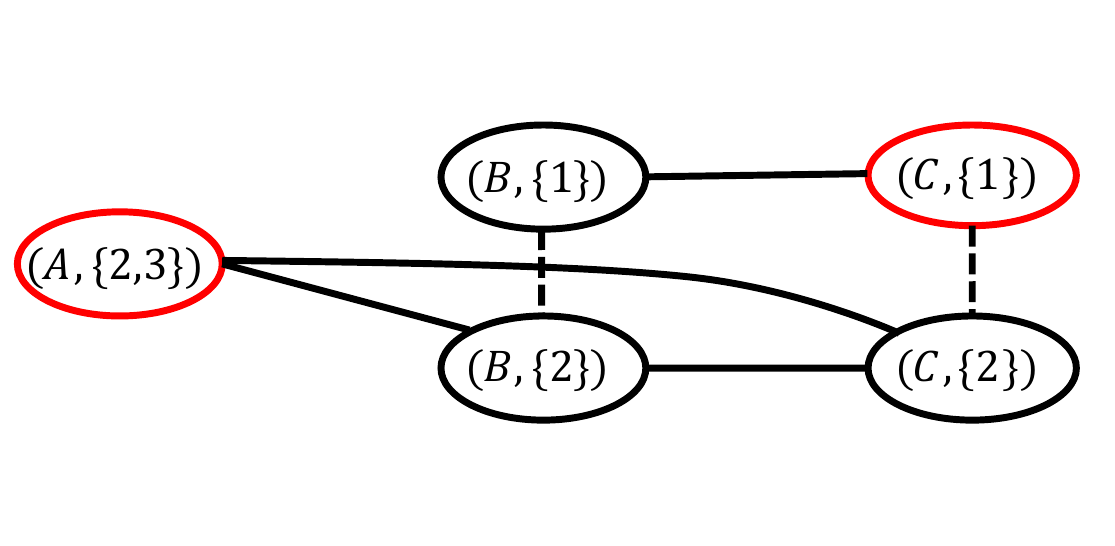}
		\label{fig:GAA_interference_graph_no_coex}}
	
	\subfloat[]{\includegraphics[width=.7\columnwidth]{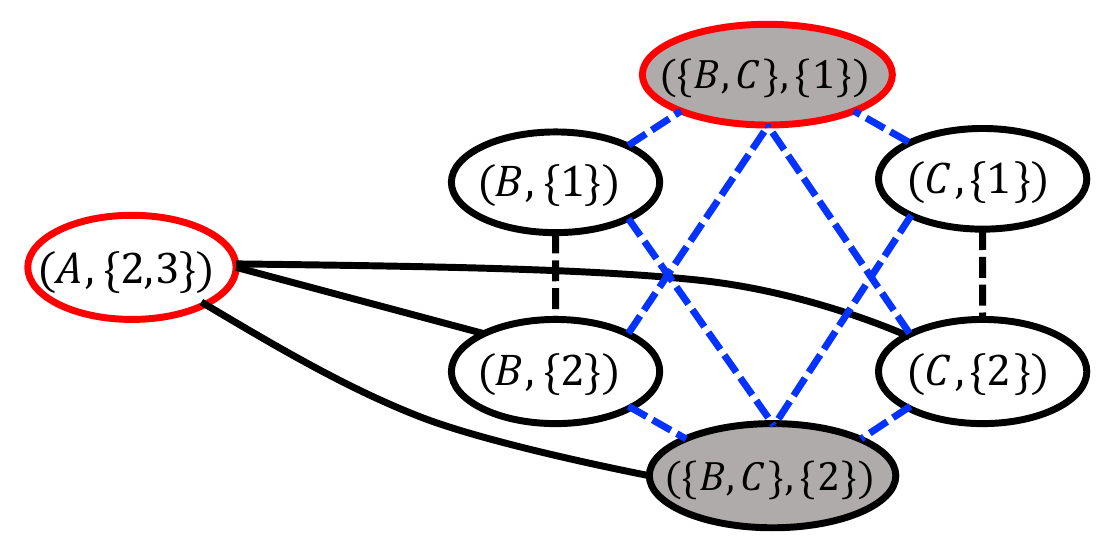}
		\label{fig:GAA_interference_graph_with_coex}}
	\caption{
		Example of (a) NC-pair conflict graph and (b) coexistence-aware NC-pair conflict graph.
		Suppose that A has two channels $\{2,3\}$ available and requests for two contiguous channel, while B and C have two available channels $\{1,2\}$ and each requests for one channel. 
		A, B, and C are interfering with each other, but B and C are within each other's CS/ED range. 
		Based on the \rev{conflict graph} in (a), the SAS may assign CH $\{2,3\}$ to A and CH $\{1\}$ to C. 
		Realizing that B and C are able to \rev{coexist under type-II conflicts, }
		the SAS may assign CH $\{1\}$ to both B and C based on the conflict graph in (b). 
	}
	\label{fig:GAA_conflict_graph_ex}
\end{figure}

\subsubsection{Rewards}
In GAA CA, \rev{the basic objectives of the SAS} would include accommodating more nodes and \rev{assigning more channels.} 
As a result, NC pairs that assign more channels to more nodes are generally more preferred. 
To reflect such preference, the SAS can define a \textit{reward (or vertex weighting) function} $R: V \mapsto \mathbb{R}^+$ that assigns a positive reward to each vertex $v=(S,C)$. 
\rev{If the SAS aims to meet more demand by assigning more channels,} the \textit{linear} reward function $R(v)=R(S,C) = |S|\cdot|C|$ may be used, which \rev{is equal to} the \rev{total} number of assigned channels \rev{at each vertex}. 
\rev{If the SAS wants to prioritize the objective of accommodating more nodes,} an alternative \textit{log} reward function \rev{may be adopted, i.e.,} $R(v)=R(S,C) = |S| \cdot ( 1 + \log(|C|))$, which captures the decreasing benefit of a node getting additional channels. 


\noindent \textbf{Problem Formulation}:
Given a vertex-weighted conflict graph $G(V,E,R)$ for $n$ nodes \rev{for binary GAA CA}\footnote{\rev{In this work, we focus on the CBSD-centric graph representation and algorithms for GAA CA. Hence, since GAA nodes are at fixed locations, the SAS does not need to consider node mobility. We will leave extending this work to client-centric and studying the impact of end user device mobility on the proposed algorithms as future work.}}, the SAS wants to find a scheme $I \subseteq V$ so as to \revv{maximize the total reward and accommodate more nodes, subject to the conflict constraint, i.e.,}
%
%
\begin{align}
\max_{I \subseteq V} &~ R'(I) \nonumber \\ \text{Subject to: } e(u,v) &\notin E, \forall u,v \in I, 
\label{eq:max_reward_CA}
\end{align}
\revv{where $R'(I)=\sum_{v\in I} (R(v) + \lambda |S(v)|)$ ($S(v)$ is the set of nodes at vertex $v$) and $\lambda \geq 0$ is the trade-off parameter chosen by the SAS.}
We refer to the problem in (\ref{eq:max_reward_CA}) as \textbf{Max-Reward CA}.


\subsection{Conflict Graph with Non-Binary Conflicts for GAA CA }
\rev{In order to} further exploit type-I conflicts, we extend the proposed conflict graph to \rev{incorporate} non-binary conflicts and formulate the CA problem as a utility maximization problem.
We start with a conflict graph without any super-NC pair as constructed in Section~\ref{sec:conflict_graph_binary}. 
To differentiate the impacts of a conflict on two interfering nodes, we consider directed edges and use a \textit{penalty (or edge weighting) function} $P: E \mapsto \mathbb{R}^+$ to assign a non-negative penalty to each directed edge.
Let the penalty of $e(u,v) \in E$  as $P(e(u,v))$ or simply $P_{u,v}$, i.e., \rev{the penalty on $v$ due to the interference from $u$.}
In general, we have $P_{u,v} \neq P_{v,u}$, and $P_{u,v}=0$ if $e(u,v) \notin E$. 


\begin{figure}[!t]
	\centering
	\includegraphics[width=.8\columnwidth]{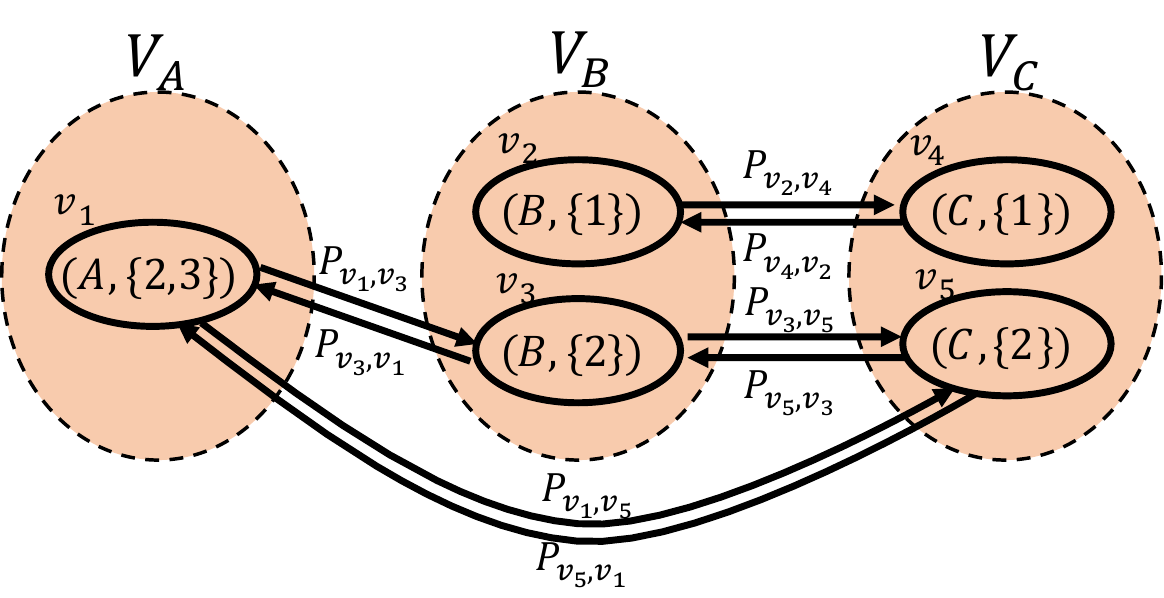}
	\caption{
		Example of NC-pair conflict graph with non-binary conflicts.
		Each vertex \rev{is associated with} a reward and each directed edge has a penalty. 
		\rev{Note that NC pairs belonging to the same node are grouped as a cluster and there are no intra-cluster edges.}
	}
	\label{fig:GAA_modified_conflict_graph}
\end{figure}


Fig.~\ref{fig:GAA_modified_conflict_graph} illustrates the proposed non-binary conflict graph.
In addition to the directed edges with penalties, the other main difference is that 
the set of NC pairs that belong to the same node are grouped as a cluster, and there are no edges within a cluster.
The one-channel-assignment-per-node constraint is enforced instead by selecting at most one NC pair from each cluster. 

\noindent \textbf{Problem Formulation}: 
Given a vertex-weighted and edge-weighted conflict graph $G(V,E,R,P)$ for $n$ nodes \rev{for non-binary GAA CA}, the SAS wants to find a scheme $I \subseteq V$ so as to 
\begin{align}
\max_{I \subseteq V} &~U(I) \nonumber \\
\text{Subject to: } |I \cap V_i| &\leq 1, \forall i=1,2,...,n, \label{eq:max_utility_CA}
\end{align}
where 
\begin{equation}
U(I) = \sum_{v \in I} R(v) - \lambda \cdot \sum_{u,v\in I} P_{u,v} \label{eq:utility_function}
\end{equation}
is the \textit{utility function} and $\{V_i\}$ are NC-pair clusters that form a partition of $V$, that is, $V=\bigcup_{i}V_i$ and $V_i\cap V_j = \emptyset$ for $i\neq j$.
We refer to the above problem as \textbf{Max-Utility CA}.

\section{Proposed Algorithms}\label{sec:algo}
In this section, we first present our super-node formation algorithm that identifies 
super-nodes 
and then present a greedy algorithm for max-cardinality and max-reward CA. 
Finally, we propose a local-search-based algorithm for max-utility CA, which exploits the structural property 
of the utility function and provides a performance guarantee. 

\subsection{Super-Node Formation}\label{sec:super_node_formation}
The super-node formation algorithm (Algorithm~\ref{algo:super_node_formation}) aims to identify a set of 
nodes that have channel $C$ available and that are within each other's CS/ED range.
As we can see, 
the main input is an undirected graph, where each vertex is a node with $C$ available, and each edge \rev{connects two} nodes \rev{that} are within each other's CS/ED range. 
\rev{Since} a super-node is a \textit{clique} (a subgraph where every two distinct vertices are connected),  
the first task is to find all such cliques.
In this work, we adopt the well-known Bron-Kerbosch algorithm \cite{bron1973algorithm},  a recursive backtracking algorithm that lists subsets of vertices that are cliques, and no listed subset can have any additional vertex without breaking its complete connectivity.
While a node can belong to multiple cliques, it randomly chooses a clique to join, since it cannot join two super-nodes at the same time.

\begin{algorithm}
\caption{Super-Node Formation}
\label{algo:super_node_formation}
\begin{algorithmic} [1]
\REQUIRE CS graph on channel $C$, activity index limit $\bar{\alpha}$
\ENSURE A set of super-NC pairs $U$ on channel $C$
\STATE Run Bron-Kerbosch algorithm to find cliques $\mathcal{Q}$.
\STATE	Each node in $V$ that belongs to multiple cliques chooses one to join (e.g., randomly).
\FORALL{clique $Q$ in $\mathcal{Q}$}
	\STATE	Run FFD to identify super-nodes $\mathcal{S}$ w.r.t. $\bar{\alpha}$.
	\STATE Add $(S,C)$ for each $S\in \mathcal{S}$ s.t. $|S|>1$ to $U$.
\ENDFOR
\RETURN $U$
\end{algorithmic}
\end{algorithm}
\normalsize

Although nodes in each clique (and any of its subsets) are able to coexist, the SAS may want to divide them into smaller groups for load balancing. 
Denote the \textit{activity index} of node $i$ as $\alpha_i > 0$, which is the estimated total number of channels it demands.
The activity index mapped to $C$ is $\alpha_i(C) = \min \left( \alpha_i/|C|, 1 \right) \in (0,1]$.
To avoid overcrowded super-nodes, the SAS can set a \textit{sum activity index limit} $\bar{\alpha}>0$ (e.g., $\bar{\alpha}=1.0$) for each super node. 
As a result, the task of grouping nodes (i.e., items with weights) into fewer super-radios (i.e., bins with capacity of $\bar{\alpha}$) becomes the well-known one dimensional \textit{bin packing problem} (BPP). 
In this work, we adopt a heuristic-based algorithm called \textit{first fit decreasing} (FFD)\cite{eilon1971loading}: it first creates a sequence of super-nodes (i.e., empty bins) and sorts nodes in non-increasing order of activity index; 
then it places each node into the lowest-indexed super-node with sufficient remaining space. 
Latest analysis has shown that FFD uses no more than $11/9$  OPT + $6/9$ bins, where OPT is the number of bins given by the optimal solution \cite{dosa2007tight}. 

\subsection{Algorithm for Max-Cardinality and Max-Reward CA}
\revv{In the max-cardinality CA (Eq.~(\ref{eq:max_cardinality_CA})), each vertex has a unit weight, while in the max-reward CA (Eq.~(\ref{eq:max_reward_CA})), each vertex has a weight equal to $(R(v)+\lambda |S(v)|)$. 
If we assume that each node has only one channel assignment, then the problems boil down to selecting an \textit{independent set} (IS), i.e., a set of vertices without any edges, that have the maximum sum of weights on general graphs, which is the classic \textit{maximum (weighted) independent set} (M(W)IS) problem.}

\revv{Since it is well known that both MIS and MWIS are NP-hard on general graphs\footnote{It has been shown that the problem may become easy on simple classes of graphs. For instance, MWIS can be solved in linear time in any tree graph. Polynomial time algorithms exist for MWIS in other classes of graphs including bipartite graphs, line graphs, circle graphs, claw-free graphs and planar graphs. See \cite{bandyapadhyay2014variant} for a brief summary.} \cite{halldorsson2000approximations,hastad1996clique}, it means that our problems are also NP-hard even when each node has only one channel assignment.} 
In this work, we adopt a heuristic-based algorithm 
called GMWIS (Algorithm~\ref{algo:GWMIS}) to obtain an approximate solution.
Denote by $\theta(G)$ the weight of a maximum IS, and by $A(G)$ the weight of the solution obtained by algorithm $A$.
The performance ratio is defined by $\rho_A = \inf_G \frac{A(G)}{\theta(G)}$.
It has been shown in \cite{sakai2003note} that GMWIS outputs an IS of weight at least $\sum_{v\in V} \frac{R(v)}{\delta_G(v)+1}$, where $\delta_G(v)$ is the vertex degree of $v$ and its performance ratio is $\frac{1}{\Delta_G}$, where $\Delta_G$ is the maximum vertex degree of $G$.
Note that the computational complexity of Algorithm~\ref{algo:GWMIS} is $O(|V|^2)$, since each iteration takes $O(|V|)$ operations and there are up to $|V|$ iterations.


\begin{algorithm}
\caption{Greedy algortihm for MWIS (GMWIS)}
\label{algo:GWMIS}
\begin{algorithmic} [1]
\REQUIRE Vertex-weighted undirected graph $G(V, E, R)$
\ENSURE A maximal independent set $I$ in $G$

\STATE Pick a vertex $v\in V$ that maximizes $\frac{R(v)}{\delta_{G}(v) + 1}$. 
\STATE Add $v$ to $I$ and remove $v$ and its neighbors from $V$.
\STATE Repeat steps 1 and 2 until all vertices in $V$ are removed.	
\RETURN $I$

\end{algorithmic}
\end{algorithm}

\vspace{-0.3cm}

\subsection{Proposed Algorithm for Max-Utility CA}
In this section, we first provide brief background on submodularity and matroids. 
Then we show that the max-utility CA problem in Eq.~(\ref{eq:max_utility_CA}) is matroid-constrained submodular maximization and proposed a polynomial-time algorithm with a provable performance guarantee.

\subsubsection{Background on Submodularity and Matroids}
Submodularity is a property of set functions that captures the \textit{diminishing returns} behavior, that is, adding a new element introduces greater incremental benefits, if there are fewer elements so far, and less, if there are more elements.
  
The formal definition is as follows.

\begin{mydef}[Submodularity]
Given a finite set $V$, a function $f:2^V\mapsto \mathbb{R}$ is submodular if, for any subsets $S,T \subseteq V$, 
\begin{equation*}
f(S) + f(T) \geq f(S\cap T) + f(S\cup T).
\end{equation*}
\end{mydef}

An equivalent definition is the following \cite{fujishige2005submodular}. 
A set function is submodular if, for any sets $S \subseteq T\subseteq V$ and any $v\in V\setminus T$, 
\begin{equation}
f(S \cup \{v\}) - f(S) \geq f(T \cup \{v\}) - f(T). \label{eq:equivalent_subomdular_def}
\end{equation}


%


\begin{mydef}[Matroid]\label{def:matroid}
Let $V$ be a finite set, and let $\mathcal{I}$ be a collection of subsets of $V$. 
A set system $\mathcal{M}=(V, \mathcal{I})$ is a matroid if the following three conditions hold: (i) $\emptyset \in \mathcal{I}$, (ii) if $B \in \mathcal{I}$, then $A \in \mathcal{I}$ for all $A \subseteq B$, and (iii) if $A,B\in \mathcal{I}$ and $|A|<|B|$, then there exists $v \in B\setminus A$ such that $(A \cup \{v\}) \in \mathcal{I}$.
\end{mydef}


Given a matroid $\mathcal{M}$, $\mathcal{I}(\mathcal{M})$ is called the set of \textit{independent sets} of $\mathcal{M}$.
A sub-class of matroids called \textit{partition matroids} is defined as follows.

\begin{mydef}[Partition Matroid]\label{def:partition_matroid}
Let $V$ be a finite set, and $V_1$, ..., $V_m$ be a partition of $V$, that is, a collection of sets such that $V_1 \cup ... V_m=V$ and $V_i\cap V_j = \emptyset$ for $i\neq j$. 
Let $k_1$, ..., $k_m$ be a collection of nonnegative integers. 
Define a set $\mathcal{I}$ by $A \in \mathcal{I}$ iff $|A \cap V_i| \leq k_i$ for all $i=1,...,m$. 
Then $\mathcal{M}=(V, \mathcal{I})$ is a matroid, called a partition matroid. 
\end{mydef}

\subsubsection{Matroid-Constrained Submodular Maximization}
As we can see, the constraint in the max-utility CA formulation in Eq.~(\ref{eq:max_utility_CA}) defines a partition matroid,
\begin{align}
\mathcal{M}=(V,~\mathcal{I}), \text{ where } \mathcal{I}&=\{I: |I \cap V_i| \leq 1, \forall i=1,2,...,n\} \nonumber \\
&\text{ and } \{V_i\} \text{ is a partition of } V, \label{eq:matroid_constraint}
\end{align}
and \rev{we can show that} the utility function in Eq.~(\ref{eq:utility_function}) is a submodular set function.  

\begin{mylemma}\label{lemma:submodular}
The utility function $U(\cdot)$ in Eq.~(\ref{eq:utility_function}) is submodular.  
\end{mylemma}
\begin{proof}
See Appendix~\ref{proof:lemma:submodular} for proof.
\end{proof} 

Hence, the problem $
\max\{U(I): I \in \mathcal{I}(\mathcal{M}) \}$ is an instance of matroid-constrained submodular maximization, \revv{which is known to be an NP-hard optimization problem \cite{lee2010maximizing}}.

\begin{algorithm}
\caption{Utility Maximization ({\tt UM})}
\label{algo:UM_algorithm}
\begin{algorithmic} [1]
\REQUIRE Partition matroid $\mathcal{M}=(V, \mathcal{I})$, utility function $U(\cdot)$, and parameter $\epsilon > 0$
\ENSURE $I$ -- Selected NC pairs

\STATE $I_1 \leftarrow$ {\tt LS} ($\mathcal{M}, U, \epsilon$)
\STATE $\mathcal{I}' \leftarrow \{I: |I \cap V'_i| \leq 1\}, \text{ where } V'_i = V_i \setminus I_1, \forall i = 1,2,...,n$
\STATE Define $\mathcal{M}'=(V', \mathcal{I}')$, where $V'=\cup_i V'_i$
\STATE $I_2 \leftarrow $ {\tt LS} ($\mathcal{M}', U, \epsilon$)
\RETURN $ I \leftarrow \arg\max_{I\in \{I_1, I_2\}} U(I)$


\algrule
\vspace{-0.1cm}
\textbf{Local Search Procedure} ({\tt LS}):
\vspace{-0.1cm}
\algrule

\REQUIRE Matroid $\mathcal{M}=(V, \mathcal{I})$, submodular function $f(\cdot)$, and parameter $\epsilon \geq 0$
\ENSURE A selected subset $I$

\STATE $I\leftarrow \emptyset$, $v \leftarrow \arg\max \{f(u)| u \in V\}$, $N \leftarrow |V|$. \label{algo:line:start}
\WHILE{$v\neq \emptyset$ and $f(I \cup \{u\}) > (1 + \frac{\epsilon}{N^2}) f(I)$}
	\STATE $I \leftarrow I \cup \{v\}$
	\STATE $v \leftarrow \arg\max \{f(I \cup \{u\}) - f(I) | u \in V\setminus I \text{ and } I \cup \{u\} \in \mathcal{I} \}$
\ENDWHILE \label{algo:line:end_of_greedy}

\WHILE{$true$}
	\IFTHEN {there exists $d \in I$ such that $f(I \setminus \{d\}) \geq (1+\frac{\epsilon}{N^2}) f(I)$} {$I \leftarrow I \setminus \{d\}$, \textbf{continue}} \label{algo:line:delete}
	\IFTHEN {there exists $a \in V \setminus I$ and $d \in I \cup \emptyset$ such that $I \setminus \{d\} \cup \{a\} \in \mathcal{I}$ and $f(I \setminus \{d\} \cup \{a\}) \geq (1 + \frac{\epsilon}{N^2}) f(I)$}
	{$I \leftarrow I \setminus \{d\} \cup \{a\}$, \textbf{continue}} \label{algo:line:add_swap}
	\STATE \textbf{break} 
\ENDWHILE
\RETURN $I$ \label{algo:line:end}
	
\end{algorithmic}
\end{algorithm}

\subsubsection{Proposed Algorithm}
Our algorithm {\tt UM} (Algorithm~\ref{algo:UM_algorithm}) is based on the (approximate) local search procedure {\tt LS} (Lines~\ref{algo:line:start}-\ref{algo:line:end}).
It \rev{quickly} finds an initial solution 
\rev{by iteratively adding the next element with the maximum incremental utility} (Lines~\ref{algo:line:start}-\ref{algo:line:end_of_greedy}) and then searches for the locally optimal solution under the local delete (Line~\ref{algo:line:delete}), add and swap (Line~\ref{algo:line:add_swap}) operations subject to the matroid constraint.
To achieve faster convergence, the parameter $\epsilon$ can be set to obtain an approximately locally optimal solution.
Note that {\tt LS} is called twice in {\tt UM}: it first obtains $I_1$ for the original matroid $\mathcal{M}$ and then $I_2$ for the new matroid $\mathcal{M}'$, which corresponds to the partition $\{V_i'\}$ of $V'$ (i.e., $V \setminus I_1$). 
The set, either $I_1$ or $I_2$, that yields a greater utility is returned. 

Theorem \ref{theorem:bound} shows that {\tt UM} (Algorithm~\ref{algo:UM_algorithm}) provides a performance guarantee. 

\begin{mytheorem}\label{theorem:bound}
Algorithm~\ref{algo:UM_algorithm} returns a solution $I$ with the following performance guarantee
\begin{equation}
U(I) \geq \frac{1}{(4+2\epsilon)} [ U(OPT) + 2 U_{\min}],
\end{equation}
where $OPT$ is the optimal solution to $\max\{U(I): I \in \mathcal{I}(\mathcal{M})\}$, and $U_{\min} = \min \{U(I): I \in \mathcal{I}(\mathcal{M})\}$. 
If $U(\cdot)$ is non-negative (i.e., $U_{\min} \geq 0$), we have $U(I)\geq \frac{1}{(4+2\epsilon)} U(OPT)$.
\end{mytheorem}
\begin{proof}
See Appendix~\ref{proof:theorem:bound} for proof. 
\end{proof}

\begin{myproposition}\label{proposition:complexity}
Algorithm~\ref{algo:UM_algorithm} is in polynomial time with runtime bounded by $O(\frac{1}{\epsilon} N^3 \log N)$. 
\end{myproposition}
\begin{proof}
See Appendix~\ref{proof:proposition:complexity} for proof. 
\end{proof}
\section{Evaluation}
\label{sec:evaluation}
We \res{implemented} the proposed algorithms in MATLAB and compare their performance against baseline algorithms through extensive simulations. 
A real-world Wi-Fi hotspot location dataset in the New York City \cite{NYCWiFiLocation} is used to simulate \rev{a dense outdoor scenario}. Our results demonstrate the advantages of the proposed graph representation and algorithms for dynamic CA in the CBRS.

\subsection{Evaluation of Max-Cardinality CA}
For PA CA, we evaluate the proposed algorithm based on GMWIS (Algorithm~\ref{algo:GWMIS}) and compare it against a baseline algorithm called Non-Preemptive Sum Multi-Coloring (npSMC) \cite{bar2000sum}.
\res{To the best of our knowledge, npSMC is the only existing algorithm that explicitly considers channel contiguity, applicable to the PC CA scenario in the CBRS band.}

\subsubsection{Baseline}
The algorithm npSMC was proposed in \cite{bar2000sum} for scheduling dependent jobs on a graph, where each vertex $v$ is a job with a length or execution time $x(v)$ (i.e., required number of colors) and each edge indicates two dependent jobs that cannot be scheduled at the same time. 
Given a certain amount of time units or colors, the objective is to color as many vertices as possible, such that each vertex is assigned a desired amount of distinct contiguous colors and adjacent vertices are assigned disjoint sets of colors.

Given a graph $G(V,E,x )$, npSMC first puts vertices of the same length to the same group by adding an edge for each pair of vertices of different lengths, that is, $E'=E\cup \{(u,v):x(u)\neq x(v)\}$. 
This ensures that at any given time, vertices from the same group are being colored. 
Then it finds a maximal independent set (e.g., Algorithm~\ref{algo:GWMIS}) in the new graph $G'=(V,E')$ and color the selected jobs to completion non-preemptively 
in each iteration, until colors are exhausted.

\subsubsection{Setup}
We consider a square area of width $m$ that consists of $m^2$ census tracts, as shown in Fig.~\ref{fig:example_tract_topo}.
We generate randomly located circular areas of radius $r_s$ and consider census tracts that overlap with a circular area as a service area for PA CA purposes.
The number of PALs \rev{held by the same PA licensee} for each service area is uniformly selected from $[1,4]$ at random. 
To simulate the worst case, we generate as many service areas as possible (up to $1000$ trials) subject to the seven-PALs-per-census-tract limit. 
Since the maximum possible size of a solution $I$ is equal to the total number of service areas $n$, we use the metric called \textit{ratio of service areas served}, $p=|I|/n$ to measure the performance for PA CA.

\begin{figure}[ht!]
	\vspace{-0.3cm}
	\centering
	\subfloat[]{\includegraphics[trim=0 .5cm 0 0,clip=true,width=.5\columnwidth]{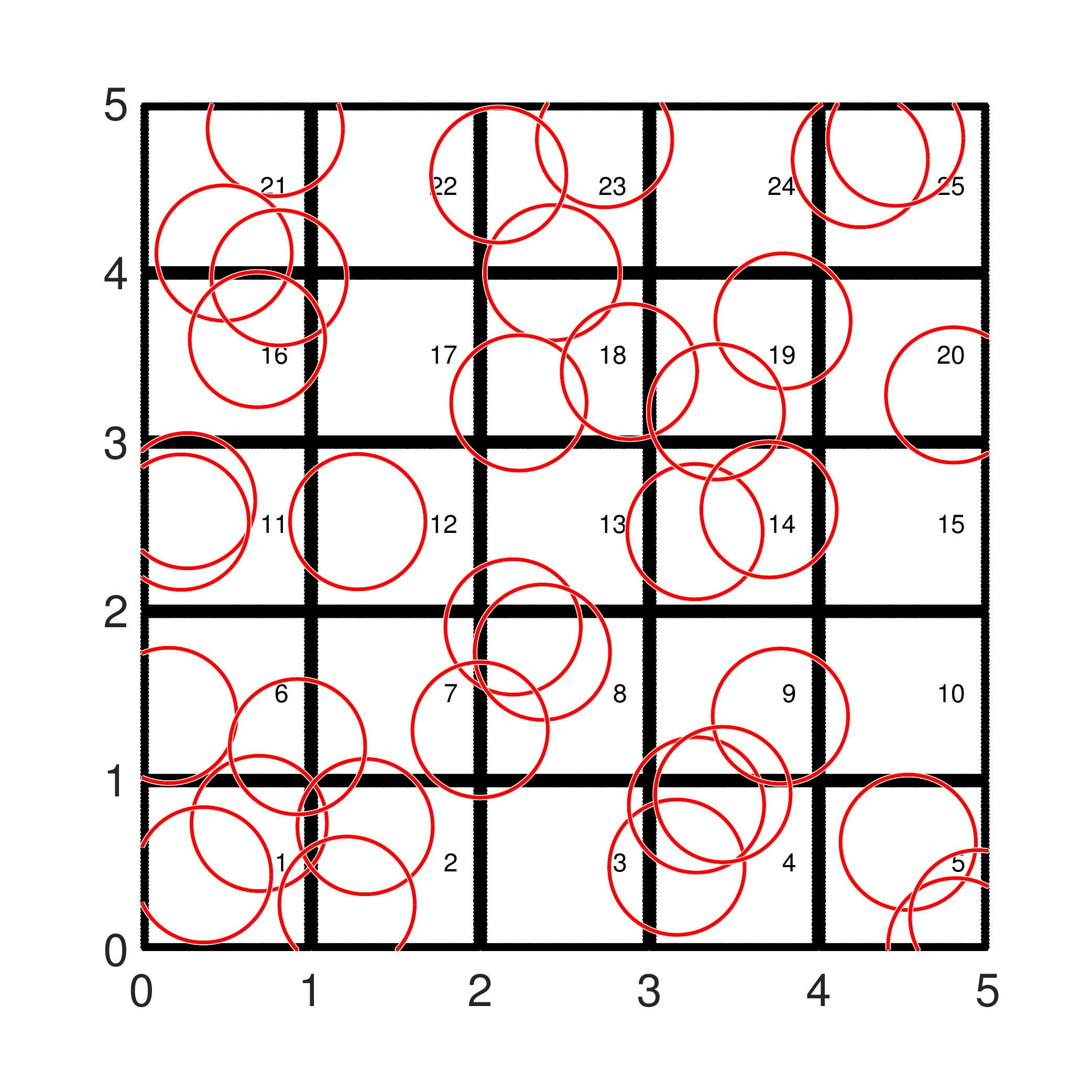}
		\label{fig:example_tract_topo}}
	\subfloat[]{\includegraphics[trim=0 -2cm 0 0,clip=true,width=.46\columnwidth]{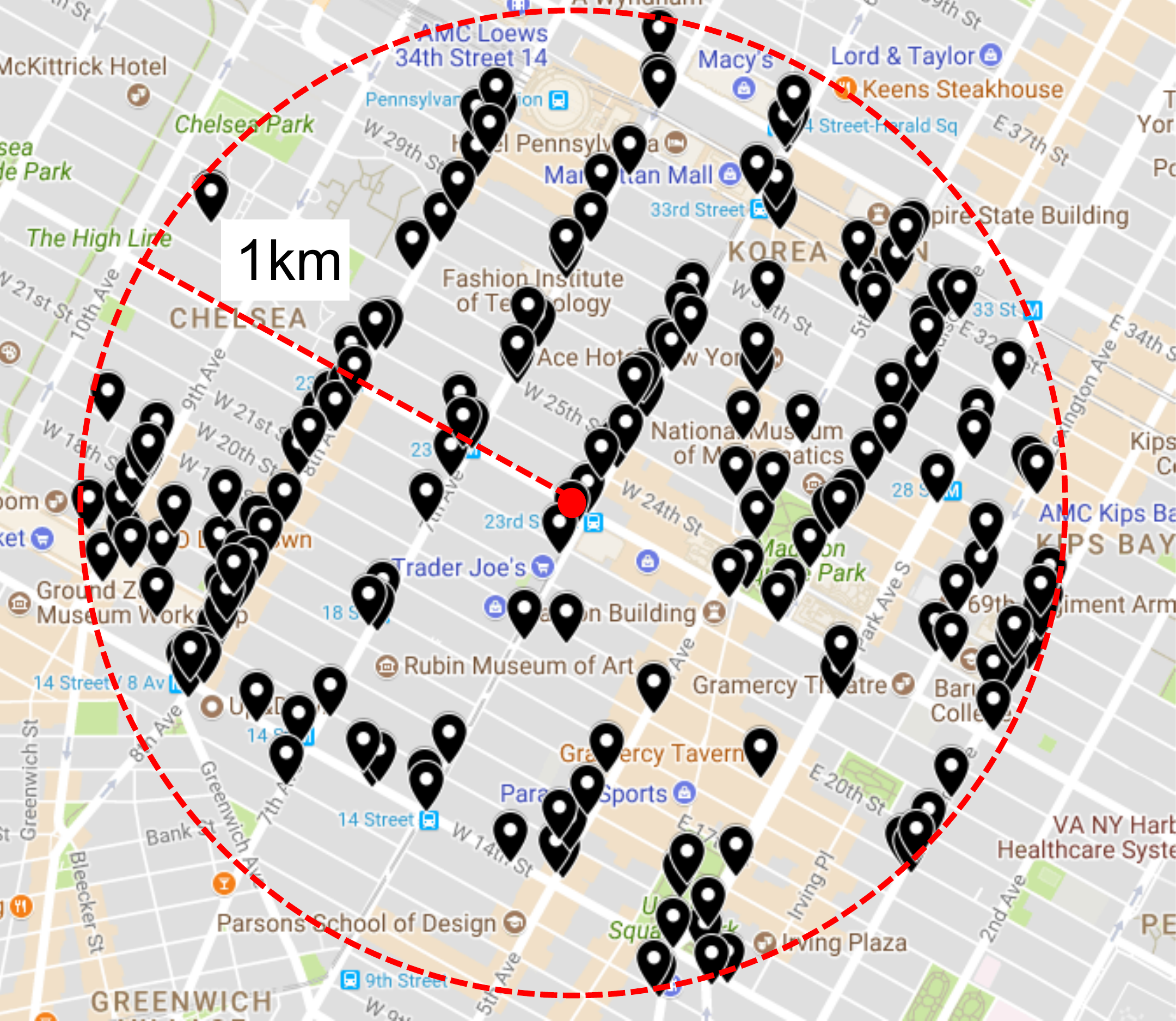}
		\label{fig:NYC_WiFi_locations}}
	\caption{(a) Example of census tracts for PA CA. The set of census tracts that overlap with a red circle are treated as a service area for a PA licensee. In this example, each circle is of radius $0.4$ and each service area contains up to four census tracts. (b) Examples of Wi-Fi hotspot locations in NYC treated as CBSDs locations for GAA CA. In this example, there are $190$ GAA CBSDs inside the circle centered at the (randomly selected) location $(40.74,-73.99)$ with a radius of $1$ km. }
\end{figure}

\subsubsection{Results}
To \rev{simulate} cities of different sizes, we set $r_s$ to $1$ and vary $m$ from $5$ to $30$ with a step of $5$.
For instance, with an optimal size of $4,000$ people for a census tract, $100$ census tracts ($m=10$) corresponds to a medium-sized city.
As mentioned in Section \ref{sec:PA_channel_assignment}, there are a total of $10$ PAL channels. 
Results are averaged over $100$ iterations. 
In each iteration, the same seed is used for the baseline and proposed algorithms for fair comparison. 

\begin{figure}[!t]
	\centering
	\subfloat[]{\includegraphics[trim=0 -.5cm 0 0,clip=true,width=.47\columnwidth]{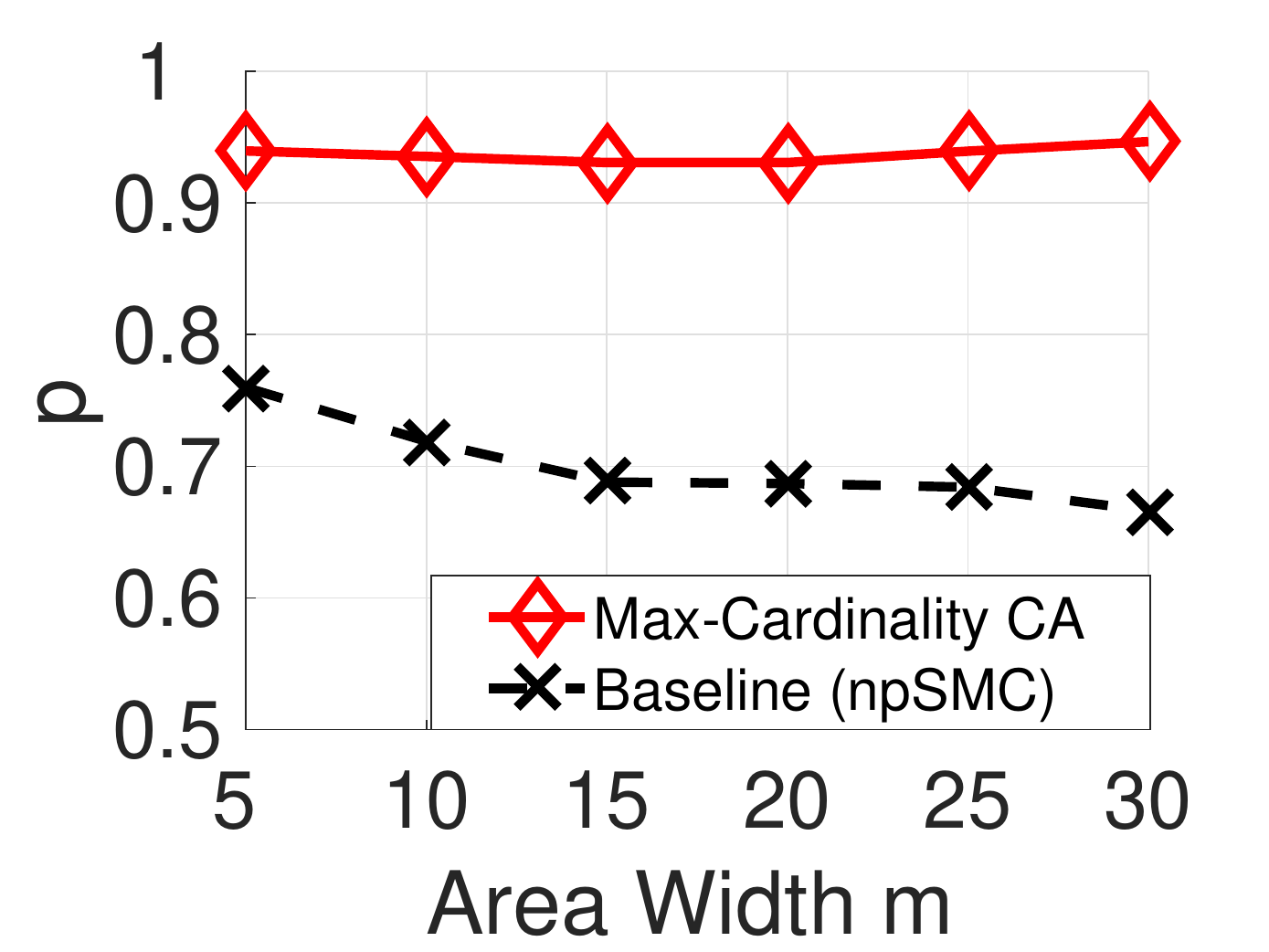}
		\label{fig:PA_CA_impact_of_area_width}}
	\subfloat[]{\includegraphics[width=.49\columnwidth]{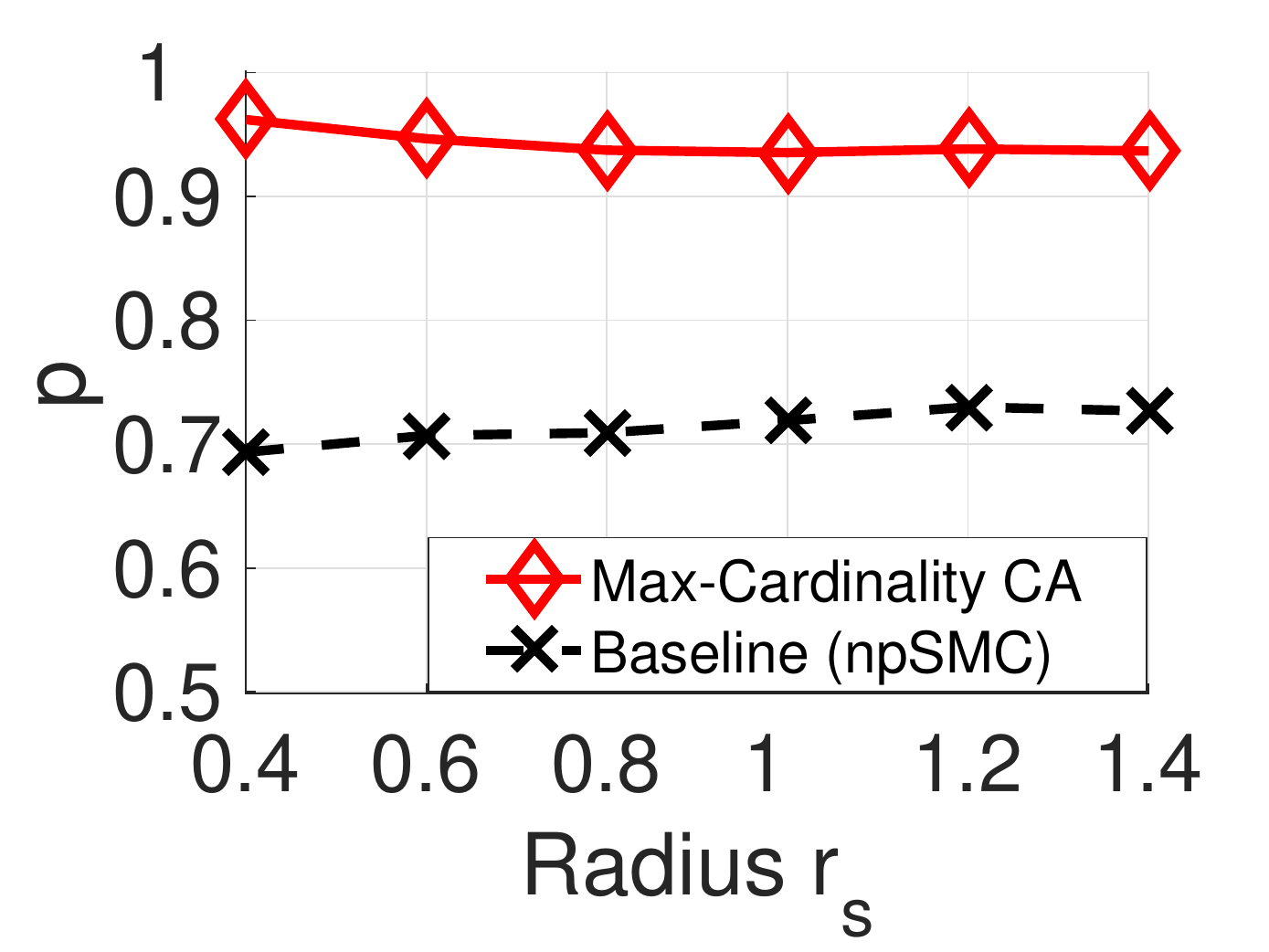}
		\label{fig:PA_CA_impact_of_radius}}
	\caption{Performance of the proposed and baseline algorithms vs. (a) area width and (b) radius. 
	\revv{The proposed algorithm consistently serves over 93.0\% of the PA demands for different area widths and radii with significant improvements over the baseline algorithm.}
	}
	\label{fig:PA_CA_result}
\end{figure}

As shown in Fig.~\ref{fig:PA_CA_impact_of_area_width}, the proposed algorithm serves $93.7\%$ of service areas on average, as compared to $70.1\%$ achieved by npSMC, with an improvement of $33.7\%$. 
We observe that due to the grouping in npSMC, service areas with larger demands tend to have lower priorities and thus often obtain no channels, whereas the proposed algorithm considers all service areas equally.

We then set $m$ to $10$ and vary $r_s$ from $0.4$ to $1.4$.
\rev{As shown in Fig.~\ref{fig:PA_CA_impact_of_radius}, the average $p$ is $94.3\%$ and $71.4\%$ for the proposed algorithm and npSMC, respectively. The improvement is as high as $32.0\%$. }
\rev{To summarize, our results show that the proposed algorithm is able to consistently meet more than $93.0\%$ of PA demands with a significant improvement over the baseline algorithm.}

\subsection{Evaluation of Max-Reward CA} \label{sec:evaluation_max_reward_CA}
In this section, we evaluate the proposed algorithm (Max-Reward) based on GMWIS (Algorithm~\ref{algo:GWMIS}) \rev{for binary GAA CA and compare it against the Max-Revenue Algorithm (MRA) in \cite{subramanian2008near}.}
Note that npSMC is not applicable, \rev{as it requires the same set of contiguous available channels at each node and cannot handle flexible demands.} 
We also evaluate the proposed algorithm 
with different settings \rev{to study the impact of reward function and coexistence awareness.}

\subsubsection{Baseline} 
\res{MRA was proposed for bidding-based spectrum allocation in \cite{subramanian2008near}. It selects a set of non-conflicting NC pairs 
so as to maximize the total revenue, which corresponds to the total reward in our case. 
It is a greedy algorithm that iteratively selects the next non-conflicting NC pair with the maximum  incremental revenue. 
Note that MRA is coexistence unaware and its performance is not affected by the choice of the reward function.}

\subsubsection{Setup}
\rev{We consider a scenario where the SAS is serving a circular region of radius $r$ (in km) that is randomly located in the densely populated Manhattan area of NYC, as illustrated in  Fig.~\ref{fig:NYC_WiFi_locations}.}
\rev{We import outdoor Wi-Fi hotspot locations within the circular region from a publicly available dataset \cite{NYCWiFiLocation} and treat them as GAA nodes.
	
In order to create location-dependent channel availability for GAA nodes, we consider two PA licensees that are assigned with CH $1$-$4$ and CH $5$-$7$, respectively, and generate $10$ randomly deployed PA nodes for each licensee.}
\rev{A channel is considered available for a GAA node, if its $-80$ dBm/$10$ MHz interference contour\footnote{Strictly speaking, $-80$ dBm/$10$ MHz is the limit for the aggregate co-channel interference, and a lower limit may be chosen for pairwise co-channel interference. Nevertheless, the choice of this limit does not affect our evaluation, since the same set of available channels are used as input for both the proposed and baseline algorithms. } does not overlap with any PA node's default PPA (Section \ref{sec:PA_channel_assignment}).}
\rev{Each GAA node has a demand set $\mathcal{D}=\{1,2,3,4\}$ and an activity index $\alpha$ in $U[0,4]$.} 
\revv{In this work, we adopt the PA interference protection rules for determining the interference relationship between two GAA nodes, that is,}
a GAA node $j$ is said to be interfering with $i$, if its interference at a location inside $i$'s $-96$ dBm/$10$ MHz service contour 
is higher than $-80$ dBm/$10$ MHz. 
The CS/ED threshold \rev{is set to $-75$ dBm/$10$ MHz (adapted from $-72$ dBm/$20$ MHz in \cite{etsi2016lte})}. 
The transmit power is $30$ dBm for each node, and transmit and receive antenna heights are $3$ and $1.5$ meters, respectively \cite{fcc2016}. 
We adopt the widely used COST-231 Hata \cite{damosso1999digital} as the propagation model.
Two metrics are adopted to measure the performance: 1) \textit{percentage of nodes served}, denoted as $p_1$,
and 2) \textit{percentage of demands served}, denoted as $p_2$, i.e., the ratio between the total number of assigned channels and the total demand. 

\begin{figure}[!t]
	\centering
	\subfloat[]{\includegraphics[width=.48\columnwidth]{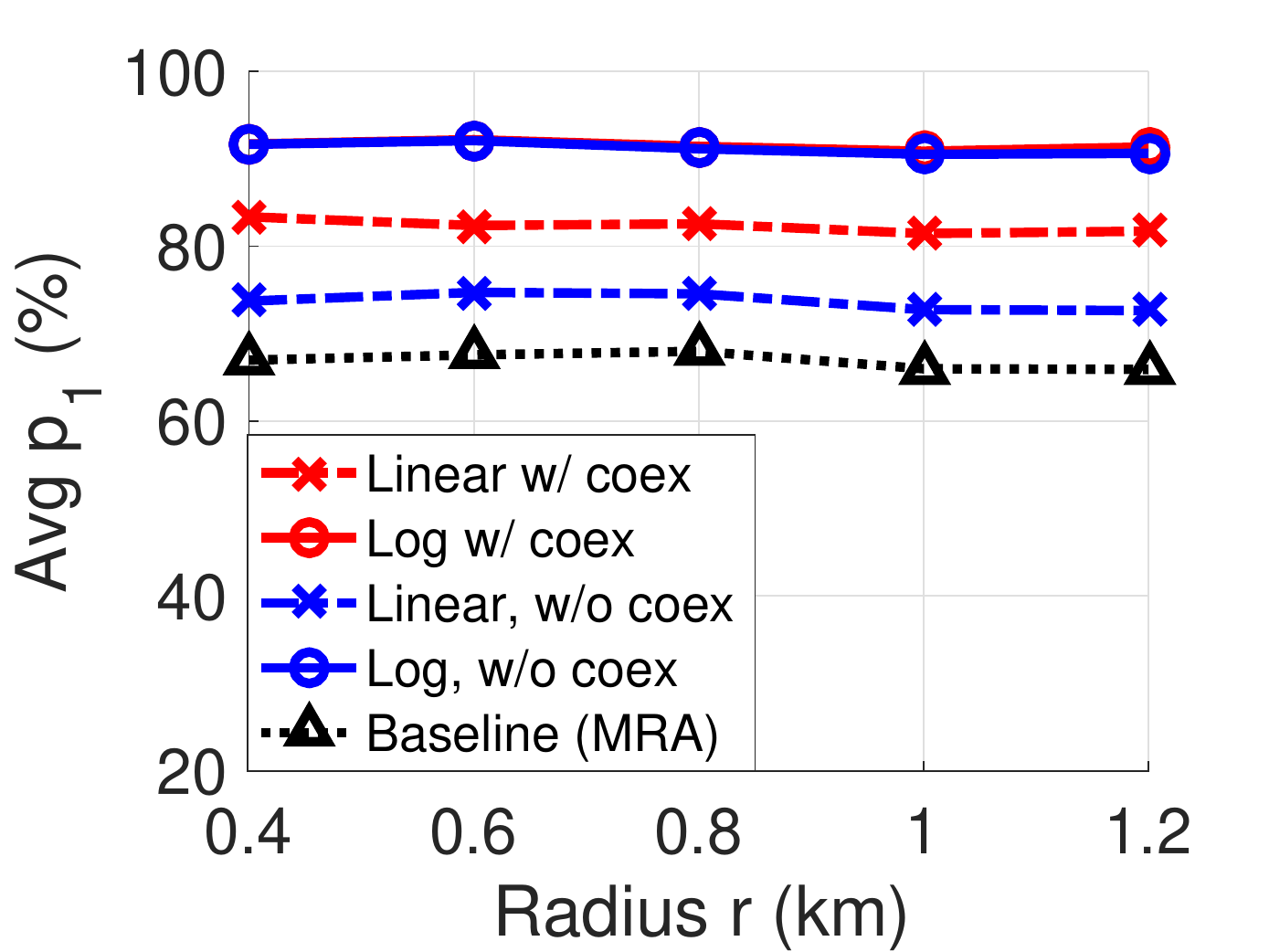}
		\label{fig:max_reward_CA_impact_of_n_p1}}
	\subfloat[]{\includegraphics[width=.48\columnwidth]{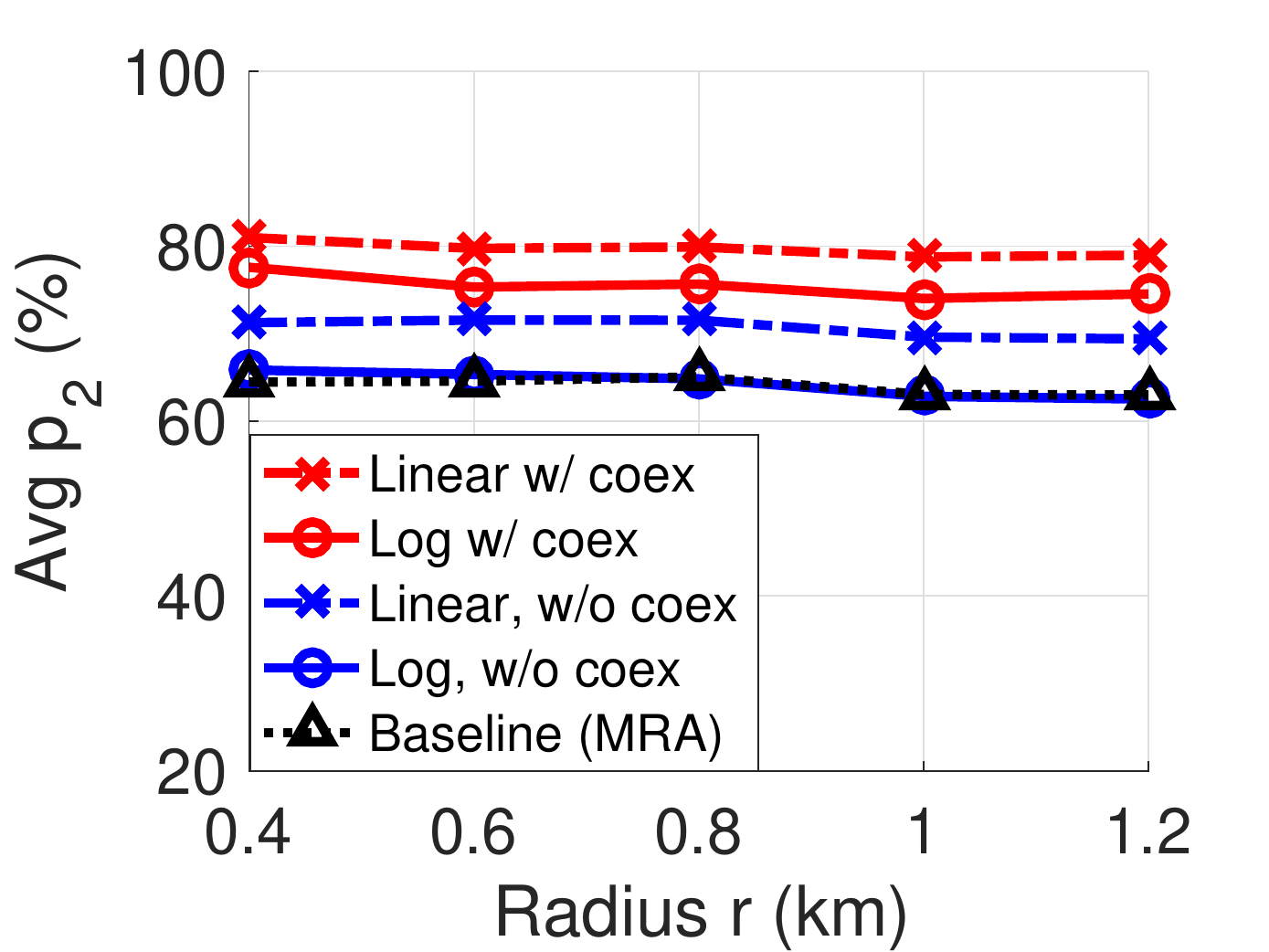}
		\label{fig:max_reward_CA_impact_of_n_p2}}
	\caption{Performance of the proposed and baseline algorithms in terms of average $p_1$ in (a) and $p_2$ in (b) as function of $n$. 
	\revv{The proposed algorithms outperform the baseline algorithm, and this advantage is further enhanced by coexistence awareness.}}
	\label{fig:max_reward_CA_impact_of_n}
\end{figure}

\subsubsection{Results}
\rev{In order to study the performance of Max-Reward for regions with different sizes, we vary $r$ from $0.4$ km to $1.2$ km.
The trade-off parameter $\lambda$ and the sum activity index limit $\bar{\alpha}$ are set to $0$ and $1$, respectively. 
Results are averaged over $30$ iterations.}

\rev{As shown in Fig.~\ref{fig:max_reward_CA_impact_of_n}, the size of the SAS's service region does not have a significant impact on the performance of a CA algorithm.
Without coexistence awareness, Max-Reward-Linear and Max-Reward-Log  achieve an average $p_1$ larger than $72.6\%$ and $90.5\%$ in all cases and outperform the baseline algorithm by $10.2\%$ and $36.4\%$ on average, respectively. 

On the other hand, Max-Reward-Log has very close performance with the baseline algorithm in terms of average $p_2$, while Max-Reward-Linear outperforms the baseline  by $10.4\%$ on average. 
The above behaviors of Max-Reward are not surprising, since the log reward function encourages assigning a channel to nodes with fewer channels, leading to a larger $p_1$ but a smaller $p_2$.
	
We can also observe that coexistence awareness effectively improves the performance  of Max-Reward-Linear in terms of average $p_1$ by $11.7\%$ but has little effects on Max-Reward-Log. On the other hand, it is able to improve both Max-Reward-Linear and Max-Reward-Log in terms of average $p_2$ by $12.8\%$ and $17.4\%$ on average, respectively. }

\begin{figure}[!ht]
	\centering
	\subfloat[]{\includegraphics[width=.48\columnwidth]{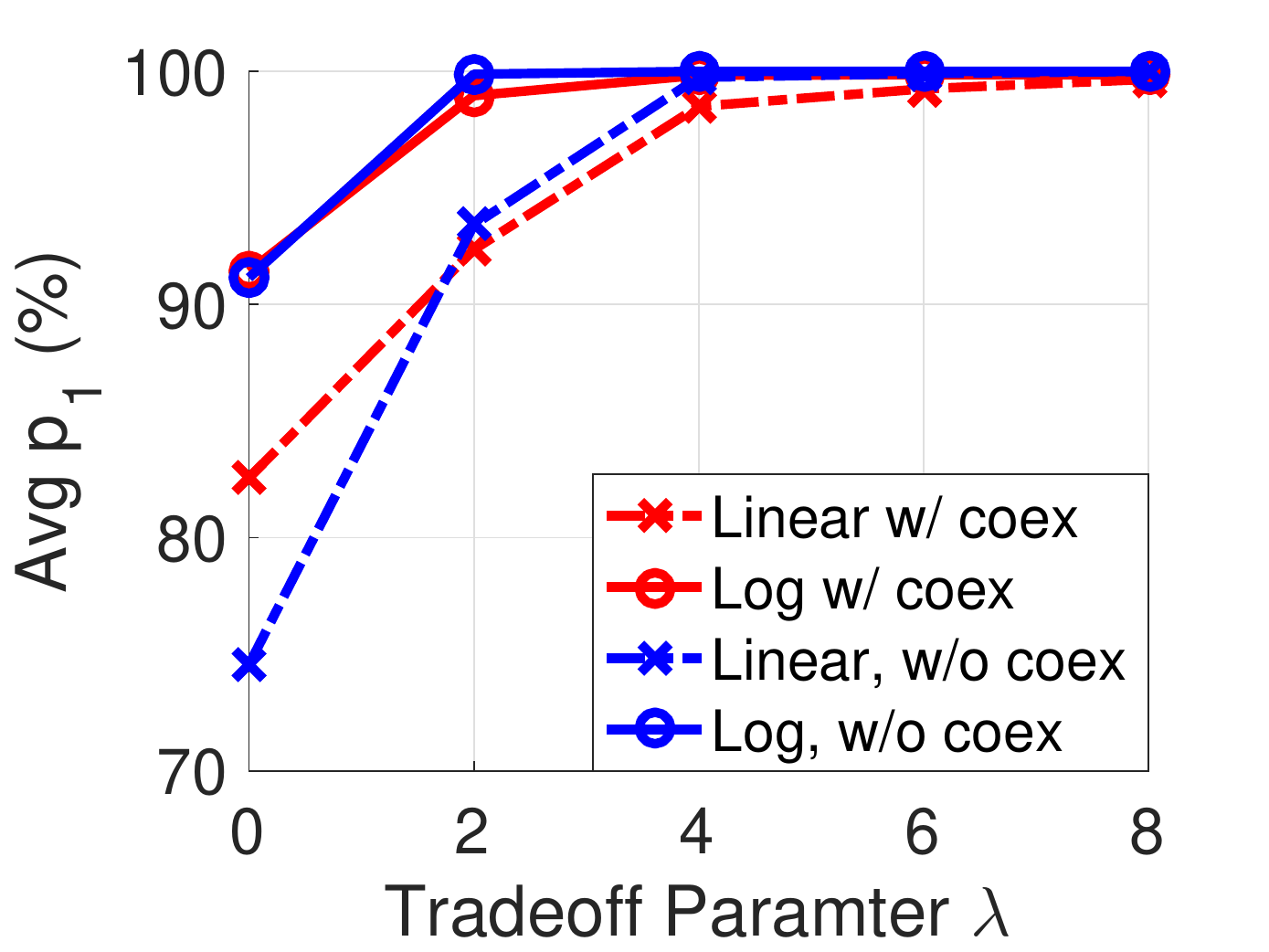}
		\label{fig:max_reward_CA_impact_of_lambda_p1}}
	\subfloat[]{\includegraphics[width=.48\columnwidth]{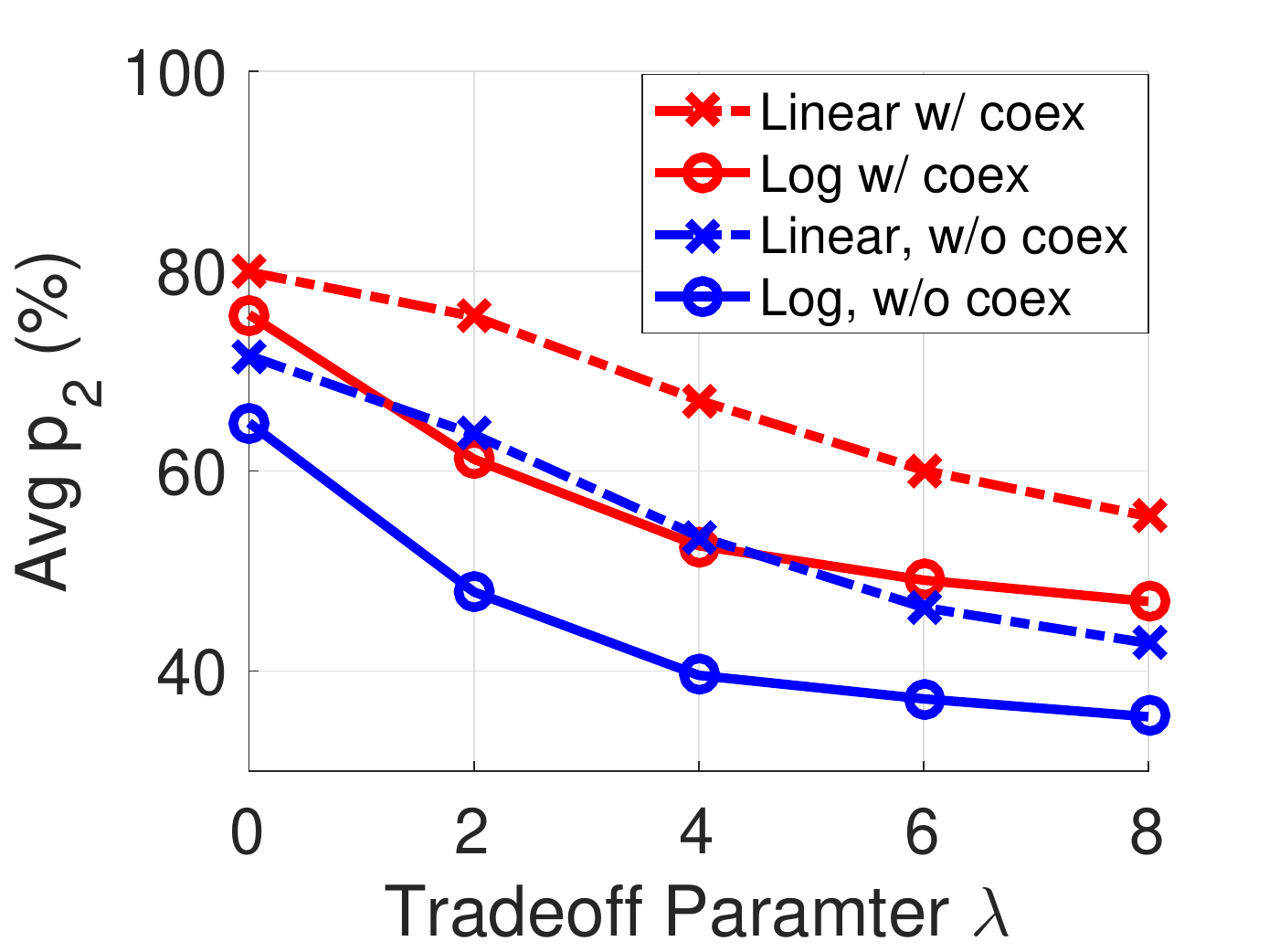}
		\label{fig:max_reward_CA_impact_of_lambda_p2}}\\
	\subfloat[]{\includegraphics[width=.48\columnwidth]{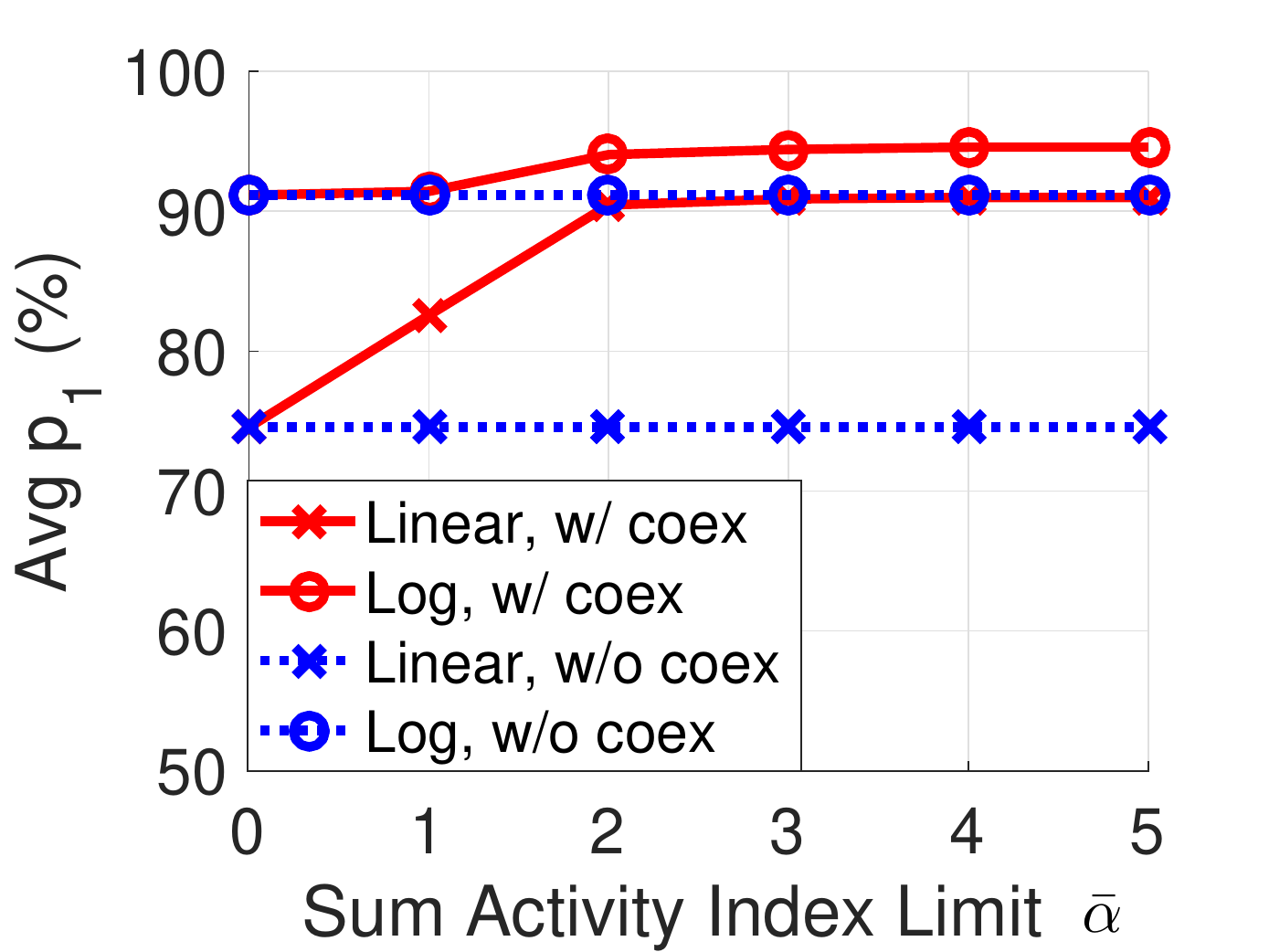} \label{fig:max_reward_CA_impact_of_alpha_limit_p1}}
	\subfloat[]{\includegraphics[width=.48\columnwidth]{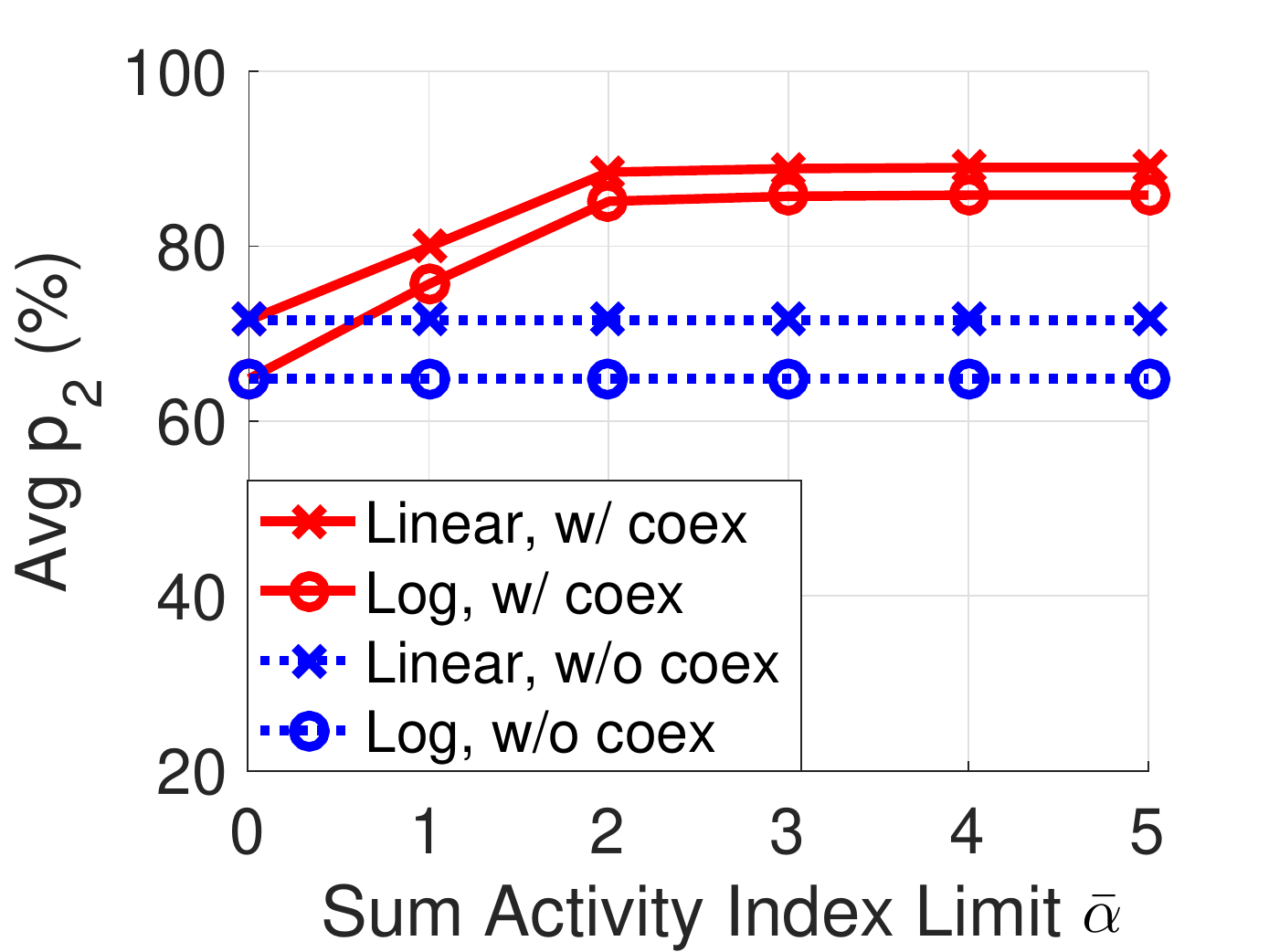} \label{fig:max_reward_CA_impact_of_alpha_limit_p2}}
	\caption{Impact of $\lambda$ and $\bar{\alpha}$ on the performance of the proposed algorithm in terms of average $p_1$ and $p_2$. 
	\revv{There exists a trade-off between $p_1$ and $p_2$ when choosing $\lambda$. 
	In addition, increasing $\bar{\alpha}$ can improve both $p_1$ and $p_2$, but such improvements no longer exist, when $\bar{\alpha}$ is large enough and all nodes in a clique are considered as a single super-node.}
	}
	\label{fig:max_reward_CA_impact_of_lambda}
\end{figure}

\subsubsection{Impact of $\lambda$ and $\bar{\alpha}$} \rev{In order to study the impact of $\lambda$, we set $r=0.8$, $\bar{\alpha}=1$, and vary $\lambda$ from $0$ to $8$. 
As illustrated in Fig.~\ref{fig:max_reward_CA_impact_of_lambda_p1} and \ref{fig:max_reward_CA_impact_of_lambda_p2}, there exists a trade-off between $p_1$ and $p_2$ when choosing $\lambda$ for the proposed algorithm.
With a larger $\lambda$, the proposed algorithm prioritizes the objective of serving more nodes by assigning more weight on the cardinality of the node set in each vertex, thus increasing $p_1$ but reducing $p_2$. 
We also see that that with a large $\lambda$, coexistence awareness does not help improve $p_1$ but is still able to effectively improve $p_2$.}

\rev{We then set $r=0.8$, $\lambda=0$ and vary $\bar{\alpha}$ from $0$ to $5$.
As shown in Fig.~\ref{fig:max_reward_CA_impact_of_alpha_limit_p1} and \ref{fig:max_reward_CA_impact_of_alpha_limit_p2}, a larger $\bar{\alpha}$ value allows more nodes to form a super-node, thus increasing both $p_1$ and $p_2$.
Nevertheless, such improvement tends to saturate at $\bar{\alpha}=2$, at which point all nodes in a clique are considered as a super-node and thus further increasing $\bar{\alpha}$ does not help any more.}

\subsection{Evaluation of Max-Utility CA}
In this section, we evaluate the proposed \revv{Max-Utility} algorithm (Algorithm~\ref{algo:UM_algorithm}) for non-binary GAA CA and compare its performance against a random selection-based baseline algorithm, which randomly selects a NC pair from each cluster. 
It repeats the process multiple times (e.g., $10000$) and returns the one with the maximum utility.

\subsubsection{Setup} 
We use the same setup in Section~\ref{sec:evaluation_max_reward_CA} and set $\epsilon=0$ for Algorithm~\ref{algo:UM_algorithm}. 
\revv{We consider two choices for reward and penalty functions.
In the first case, the number of channels assigned to each node is used as the reward, which is proportional to the best-case capacity 
(no interference). As in \cite{subramanian2007fast, kim2015design}, the estimated interference 
is used as the penalty  (normalized by the maximum pairwise interference).

In the second case, we are interested in the overall capacity of GAA networks \cite{hessar2015capacity}. Suppose that vertices $v$ and $u$ are selected for nodes $i$ and $j$, respectively. 
Let the signal to noise and interference ratio of node $i$'s client at location $x$ in channel $c_k$ be $\text{SINR}_i(x,c_k)=\gamma_i^{(k)}(x)/(N_0 W_0 +\sum_{j\neq i} \gamma_{i,j}^{(k)}(x))$, where $N_0$ is the noise power density (Watts/Hz), $W_0=10$ MHz is the bandwidth of a GAA channel, and $\gamma_i^{(k)}(x)$ and $\gamma_{i,j}^{(k)}(x)$ are the received power (Watts) and the interference due to node $j$ (Watts) at location $x$ in channel $c_k$, respectively.  
Hence, the capacity in channel $c_k$ averaged over the service area $\mathcal{A}_i$ is $\bar{C}_i(c_k)= W_0 \int_{\mathcal{A}_i} \frac{\log_2 (1+ \text{SINR}_i(x, c_k))}{\mathcal{A}_i} d x$, and the overall capacity of a GAA network is the sum of capacities over all assigned channels, i.e., $\bar{C}_i = \sum_{c_k \in C(v)} \bar{C}_i(c_k)$, where $C(v)$ is the set of assigned channels at vertex $v$. 

We consider the best-case capacity (no interference) as the reward of vertex $v$, namely, $R(v)=\bar{C}_i$, where $\bar{C}_i=\sum_{c_k \in C(v)} W_0 \int_{\mathcal{A}_i} \frac{\log_2 (1+ \text{SINR}_i(x, c_k))}{\mathcal{A}_i} d x$ and $\text{SINR}_i(x, c_k)=\gamma_i^{(k)}(x)/(N_0 W_0)$.  
If vertex $u$ is also selected, the SINR at location $x$ will be $\text{SINR}_i(x,c_k)=\gamma_i^{(k)}(x)/(N_0 W_0 + \gamma_{i,j}^{(k)}(x))$ and let the resulting capacity be $\bar{C}_i'$ \footnote{We focus on the pairwise interference between GAA nodes and thus omit the interference from incumbents.}. 
The penalty on edge $e(u,v)$ is defined as the capacity reduction, i.e., $P(e(u,v))=\bar{C}_i-\bar{C}'_i$. 
Therefore, the utility (with $\lambda=1$) corresponds to the overall capacity under pairwise interference.
Nevertheless, the parameter $\lambda$ 
can be adjusted to account for the difference between the actual capacity and that under pairwise interference.}

\begin{figure}[!t]
	\centering
	\subfloat[]{\includegraphics[width=.48\columnwidth]{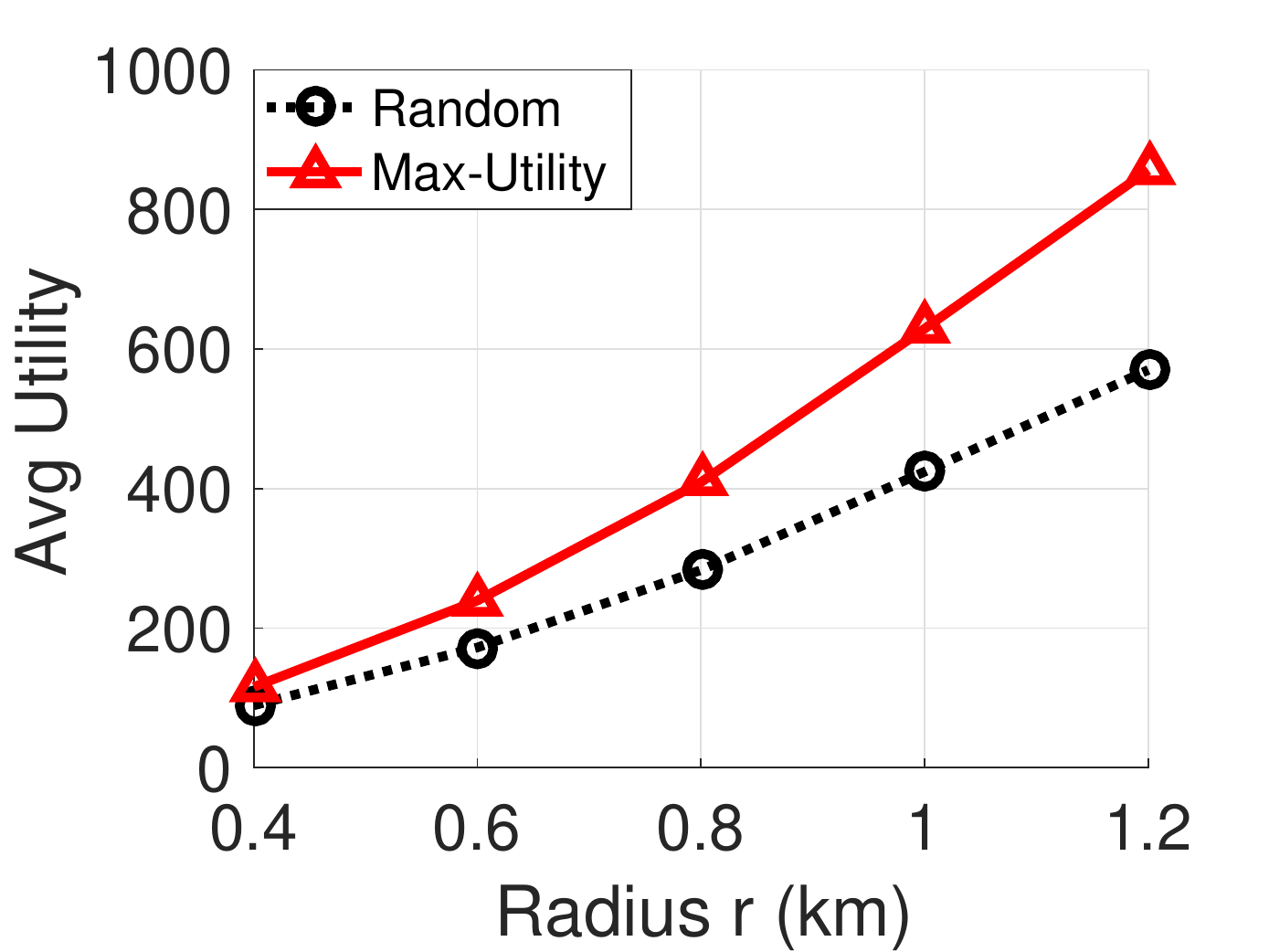}\label{fig:max_utility_CA_impact_of_radius_utility}}
	\subfloat[]{\includegraphics[width=.49\columnwidth]{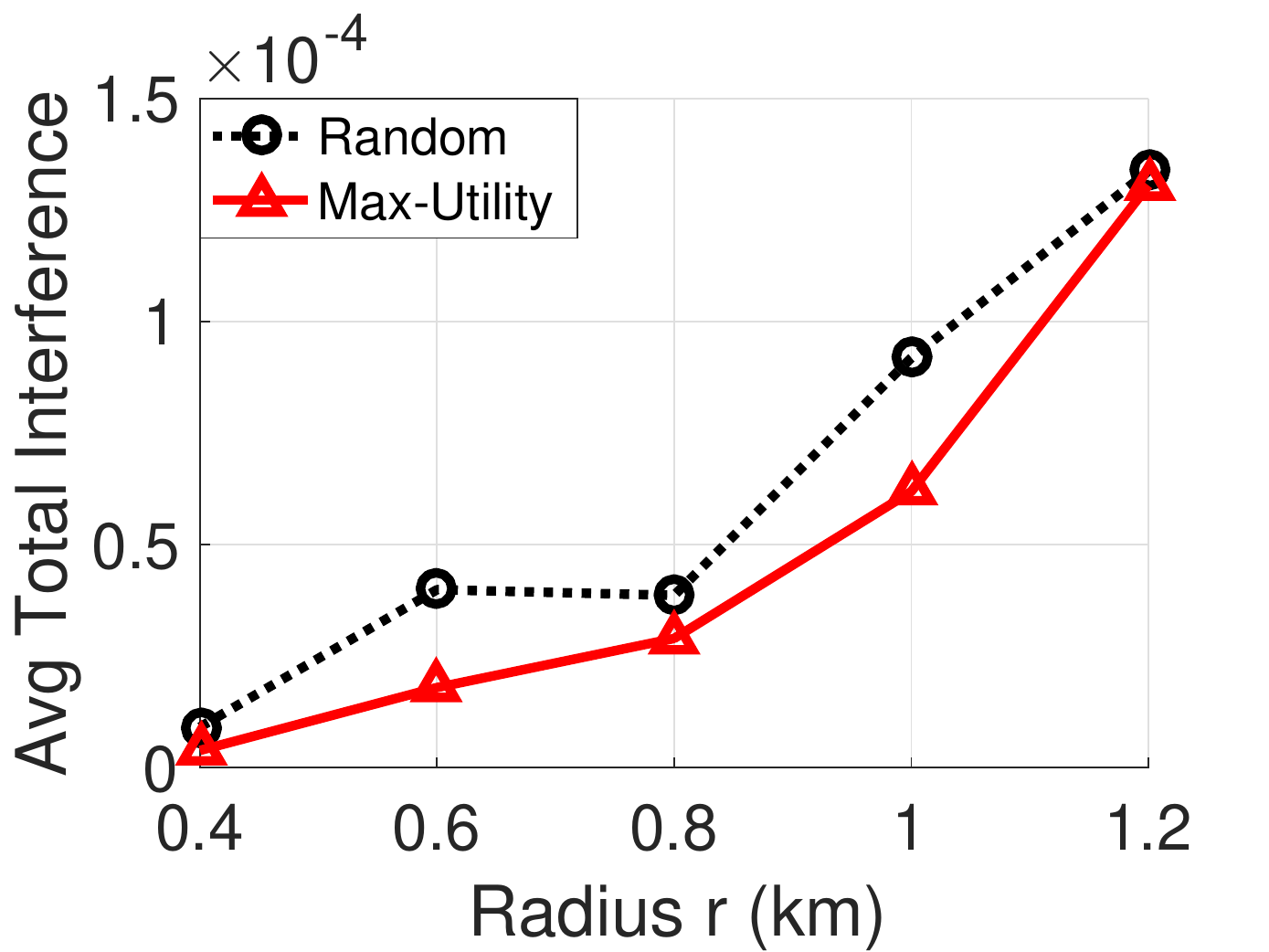}\label{fig:max_utility_CA_impact_of_radius_interference}}
	\caption{Performance of \revv{Max-Utility and random selection} in terms of (a) averaged utility and (b) averaged total interference as a function of the radius of SAS's service region. \revv{Max-Utility achieves a much greater utility and smaller total interference than random selection.}}
	\label{fig:max_utility_CA_impact_of_radius}
\end{figure}

\subsubsection{Results}
\revv{In the first case,} we set $\lambda=1$ and vary the radius $r$ from $0.4$ km to $1.2$ km. Results averaged over $30$ iterations are provided in Fig.~\ref{fig:max_utility_CA_impact_of_radius}.
As we can see, \revv{Max-Utility} is able to achieve a much greater utility than \revv{random selection}, especially when SAS's service region becomes larger. The improvement in average utility is more than $29.5\%$ in all cases and reaches $50.2\%$ with $r=1.2$ km. 
Besides, \revv{Max-Utility generates smaller total interference than random selection and allows the SAS to serve almost all nodes and demands.}

\revv{The impact of $\lambda$ is illustrated in Fig.~\ref{fig:max_utility_CA_impact_of_lambda} with $r=0.8$ km.}
We can observe that increasing $\lambda$ reduces the achieved utility, but the decrease for \revv{Max-Utility} is slower than \revv{random selection}. This is because a larger $\lambda$ implies a larger weight on the interference, which effectively makes \revv{Max-Utility} avoid selecting NC pairs with larger interference.
It also explains the decrease in the total interference for \revv{Max-Utility}. 
In all cases, \revv{Max-Utility} achieves a much greater utility while keeping the total interference much smaller than \revv{random selection}.

\begin{figure}[!t]
	\centering
	\subfloat[]{\includegraphics[width=.48\columnwidth]{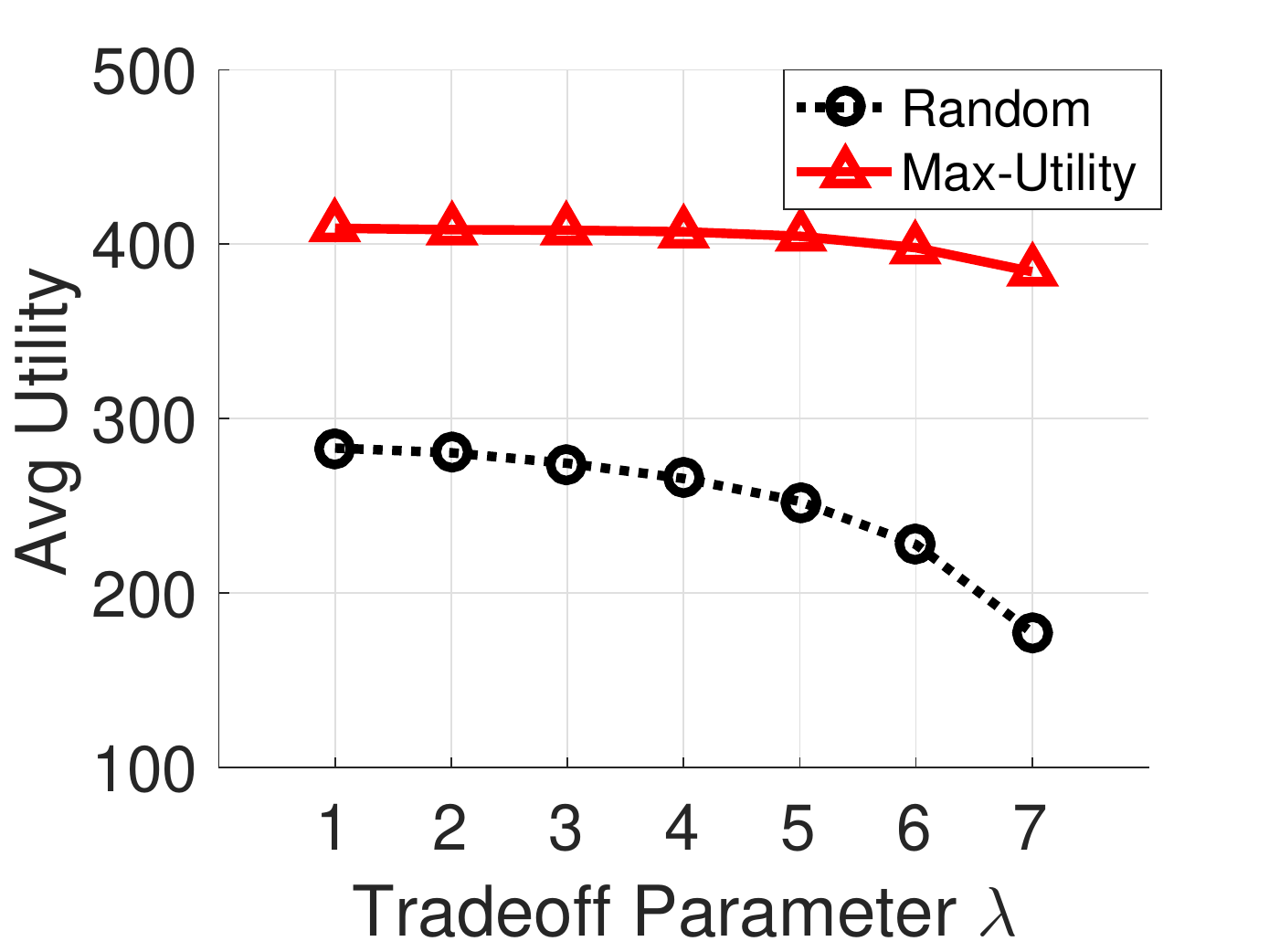}\label{fig:max_utility_CA_impact_of_lambda_utility}}
	\subfloat[]{\includegraphics[width=.48\columnwidth]{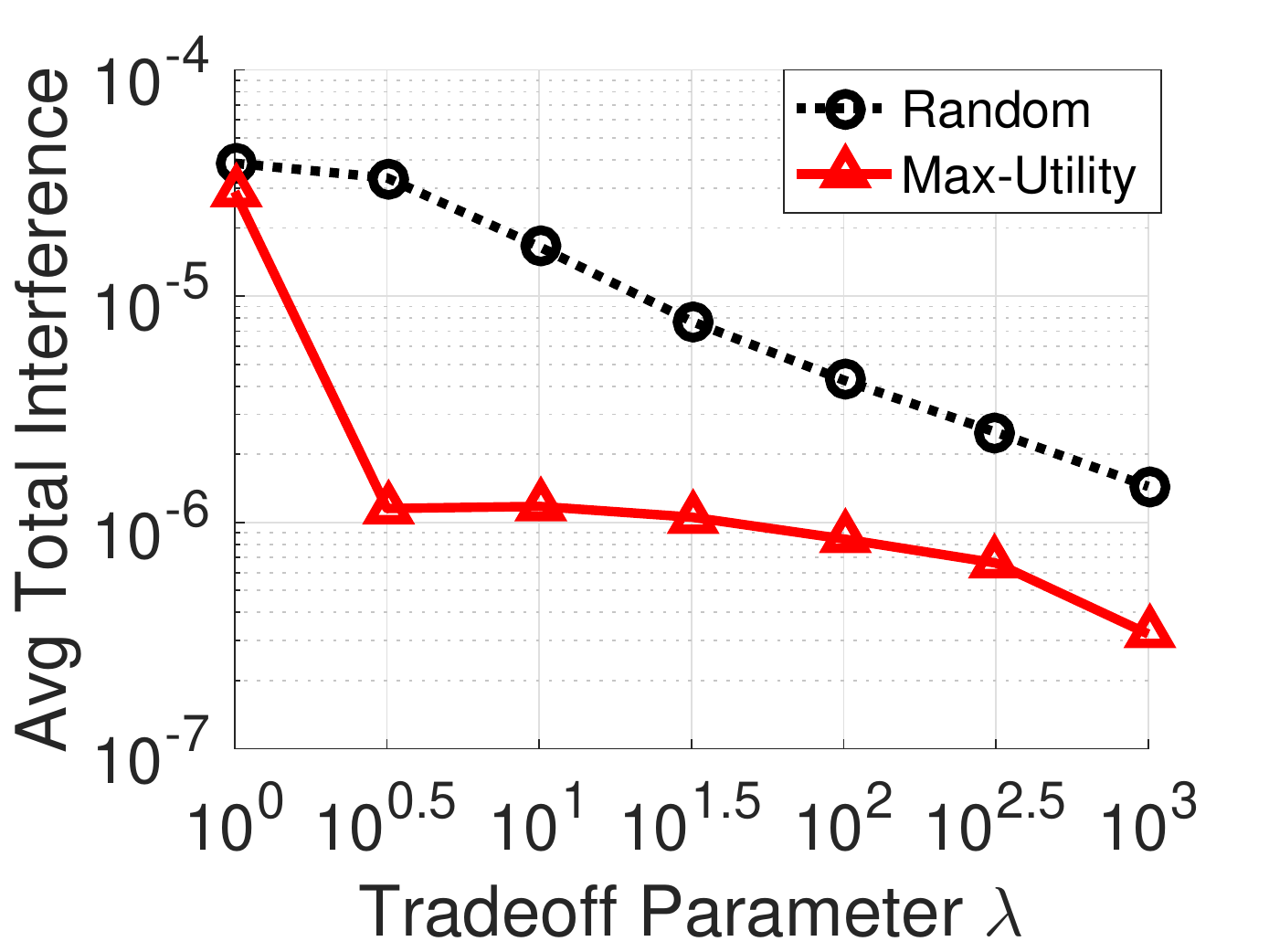}\label{fig:max_utility_CA_impact_of_lambda_interference}}
	\caption{Performance of \revv{Max-Utility and random selection} in terms of (a) averaged utility and (b) averaged total interference as a function of the trade-off parameter $\lambda$. \revv{We can see that increasing $\lambda$ puts more weights on the interference, thus leading to less interference at the cost of reduced utility.}}
	\label{fig:max_utility_CA_impact_of_lambda}
\end{figure}

\begin{figure}[!h]
	\centering
	\subfloat[]{\includegraphics[width=.48\columnwidth]{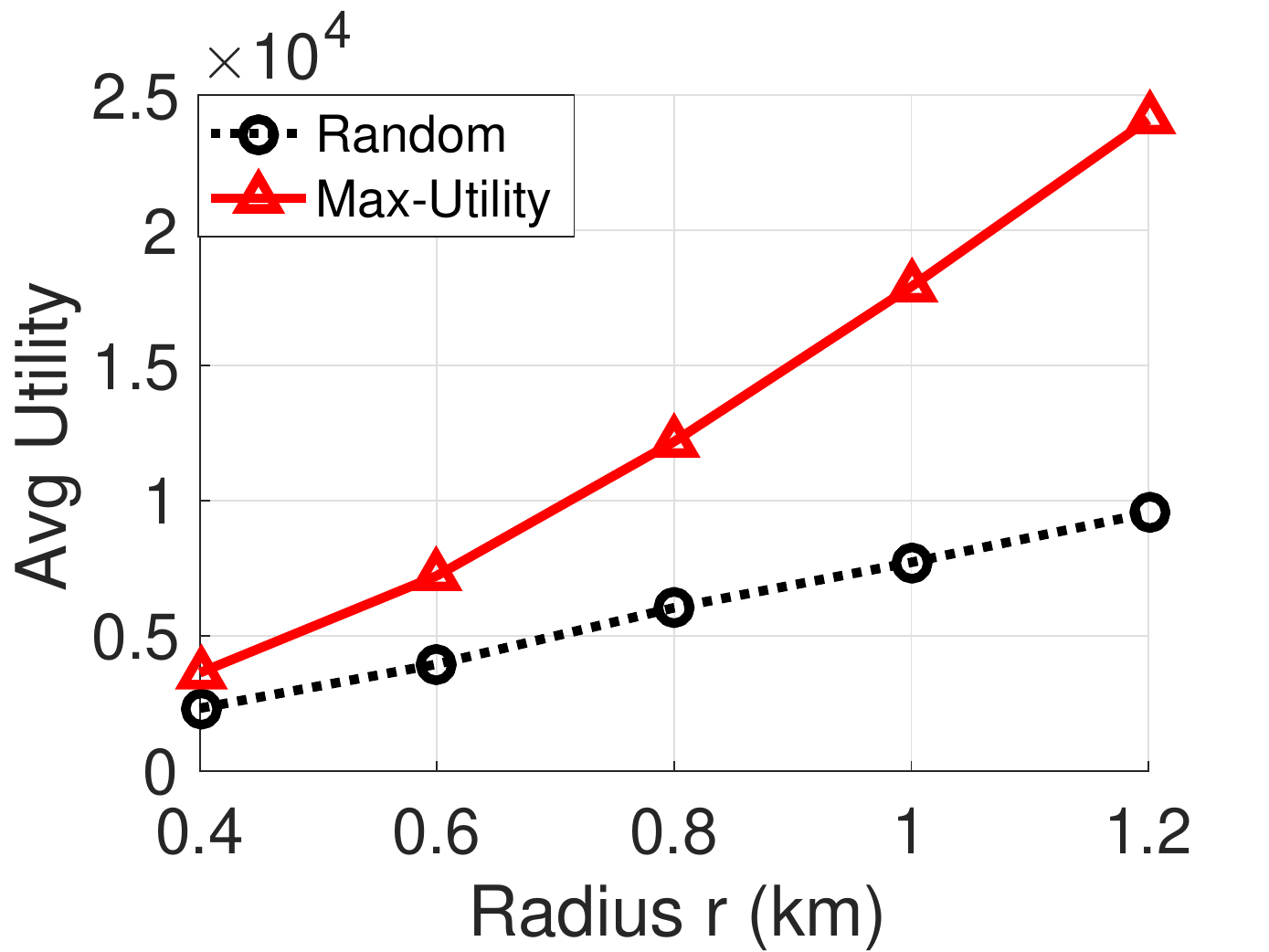}\label{fig:max_utility_CA_impact_of_radius_capacity}}
	\subfloat[]{\includegraphics[width=.48\columnwidth]{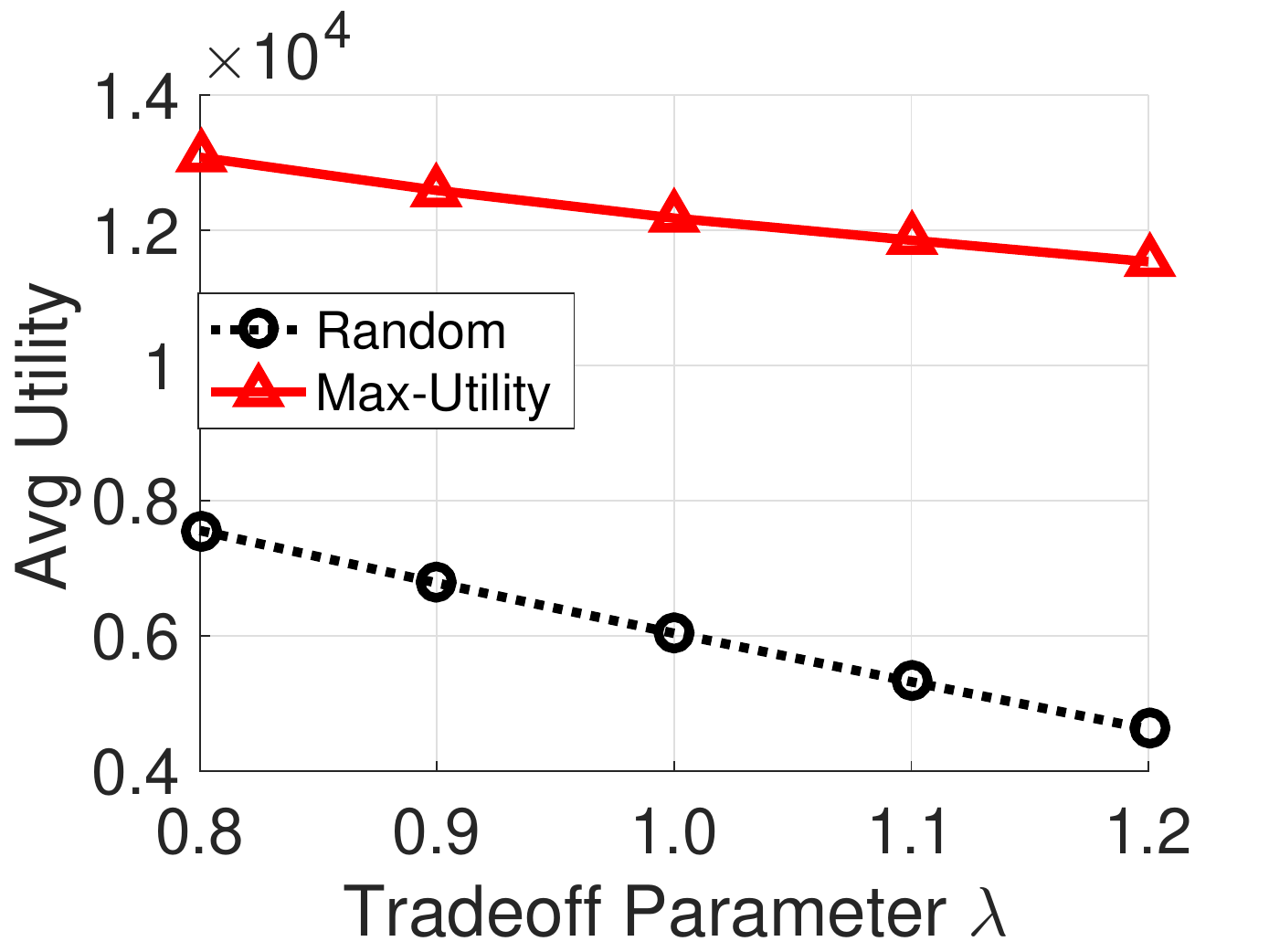}\label{fig:max_utility_CA_impact_of_lambda_capacity}}
	\caption{\revv{Average utility or capacity (Mbits/sec) achieved by Max-Utility and random selection as a function of the radius  $r$ and trade-off parameter $\gamma$. In both settings, Max-Utility leads to a much higher utility than random selection. }}
	\label{fig:max_utility_CA_capacity}
\end{figure}

\revv{In the second case, the same experiments are repeated and results are provided in Fig.~\ref{fig:max_utility_CA_capacity}. 
We observe that as $r$ increases, the average utility or capacity achieved by Max-Utility increases at a much faster rate than random selection. 
On the other hand, increasing $\lambda$ also increases the achieved utility, but Max-Utility significantly outperforms random selection.}

\section{Conclusion}\label{sec:conclusion}
In this paper, we \rev{studied SAS-assisted dynamic CA for the 3.5 GHz band.}
We proposed NC-pair conflict graph to model pairwise interference, spatially varying channel availability and channel contiguity. 
We further introduced super-NC pairs to exploit coexistence opportunities under type-II conflict and proposed the super-node formation algorithm to identify super-nodes.
The proposed conflict graph enables us to formulate PA CA as max-cardinality CA and GAA CA with binary conflicts as max-reward CA.
We adopted a heuristic-based algorithm to obtain approximate solutions.
To further enhance coexistence awareness, we extended binary conflicts to non-binary by assigning each conflict a penalty.
Then we formulated non-binary GAA CA as utility maximization. 
We showed that the utility function is submodular and  our problem is an instance of matroid-constrained submodular maximization. 
We proposed a local-search-based polynomial-time algorithm for max-utility CA that provides a provable performance guarantee. 

Our simulation results based on a real-world dataset show that the proposed max-cardinality algorithm consistently serves over $93.0\%$ of service areas and outperforms the baseline algorithm by over $30.0\%$ for PA CA.
For binary GAA CA, the proposed max-reward algorithm with linear rewards can accommodate $10.2\%$ more nodes and serve $10.4\%$ more demand on average than the baseline, and coexistence awareness can effectively improve the performance of the proposed algorithm.
For non-binary GAA CA, the proposed max-utility algorithm achieves $29.5\%$ more utility with much less total interference as compared to \revv{random selection}.
\revv{We also show that an operator can leverage the proposed framework to optimize the overall capacity of GAA networks.}

\bibliographystyle{IEEEtran}
\bibliography{IEEEabrv}

\begin{appendix}
	\subsection{Proof of Lemma~\ref{lemma:submodular}} \label{proof:lemma:submodular}
\begin{proof}
	Consider any two sets $S$, $T$ with $S\subseteq T\subseteq V$ and any $v \in V\setminus T$.
	Since $P_{u,v} \geq 0$, we any $u,v \in V$, we have
	\begin{align*}
	U(S &\cup \{v\}) - U(S) = R(v) - \lambda \cdot \sum_{\substack{u \in S}} [P_{u,v} + P_{v,u}] \\
	&\geq R(v) - \lambda \cdot ( \sum_{\substack{u \in S}} [P_{u,v} + P_{v,u}] + \sum_{\substack{u \in T\setminus S}} [P_{u,v} + P_{v,u}] ) \\
	&= R(v) - \lambda \cdot \sum_{\substack{u \in T }} [P_{u,v} + P_{v,u}] \geq U(T \cup \{v\}) - U(T),
	\end{align*}
	\normalsize
	which establishes the submodularity of $U(\cdot)$ by the definition in  Eq.~(\ref{eq:equivalent_subomdular_def}).
\end{proof} 

\subsection{Proof of Theorem~\ref{theorem:bound}} \label{proof:theorem:bound}
\begin{proof}
	Let $OPT_1 = OPT \cap V$, and $OPT_2 = OPT \cap V'$.
	Hence, we have $OPT_1 = OPT$.
	It has been shown in \cite{lee2010maximizing} that the {\tt LS} procedure returns an approximately locally optimal solution $I$ such that  $(2+\epsilon) f(I) \geq f(I\cup C) + f(I \cap C)$ for any $C \in \mathcal{I}$, where $\epsilon \geq 0$ is the parameter used in {\tt LS}.  
	The above result implies that
	\begin{align*}
	(2 + \epsilon) U(I_1) &\geq U(I_1 \cup OPT_1) + U(I_1 \cap OPT_1) \\
	(2 + \epsilon) U(I_2) &\geq U(I_2 \cup OPT_2) + U(I_2 \cap OPT_2).
	\end{align*}
	Using $U(I) \geq \max\{U(I_1), U(I_2)\}$, we have
	\begin{align*}
	(4&+2\epsilon) U(I) \geq [U(I_1 \cup OPT_1) + U(I_2 \cup OPT_2)] \\
	&~~~~+ U(I_1 \cap OPT_1) + U(I_2 \cap OPT_2) \\
	&\geq U(I_1 \cup I_2 \cup OPT_1) + [U(OPT_2) + U(I_1 \cap OPT_1)] \\
	&~~~~+ U(I_2 \cap OPT_2) \\
	&\geq U(I_1 \cup I_2 \cup OPT_1) + U(OPT_1) + U(I_2 \cap OPT_2) \\
	&\geq U(OPT_1) + 2 U_{\min} = U(OPT) + 2 U_{\min}.
	\end{align*}
	\normalsize
	The first inequality is obvious.
	The second inequality follows from submodularity, using $(I_1 \cup OPT_1) \cup (I_2 \cup OPT_2) = I_1 \cup I_2 \cup OPT_1$ and $(I_1 \cup OPT_1) \cap (I_2 \cup OPT_2) = OPT_2$.
	The third inequality is also follows from submodularity, using $OPT_2 \cup (I_1 \cap OPT_1) = OPT_1$ and $OPT_2 \cap (I_1 \cap OPT_1) = \emptyset$.
	Hence, we have $U(I) \geq \frac{1}{(4 + 2\epsilon)} [U(OPT) + 2 U_{\min}]$.
\end{proof}

\subsection{Proof of Proposition~\ref{proposition:complexity}}\label{proof:proposition:complexity}
	\begin{proof}
		Let the first element of $I$ be $v_1$ (Line~\ref{algo:line:start} in {\tt LS}), and we have 
		\begin{equation*}
		f(OPT) \leq \sum_{s \in OPT} f(s) \leq |OPT|\cdot f(s^*) \leq N \cdot f(s^*) \leq N \cdot f(v_1)
		\end{equation*}
		where $s^* = \arg \max_{s \in OPT} f(s)$ and the first inequality is due to  submodularity. 
		Since $f(v_1) \geq \frac{f(OPT)}{N}$ and each local operation increases the value by a factor $(1+\frac{\epsilon}{N^2})$, the maximum number of iterations is $\log_{1+\frac{\epsilon}{N^2}} \frac{f(OPT)}{\frac{f(OPT)}{N}} = O(\frac{1}{\epsilon} N^2 \log N)$. 
		As each iteration requires at most $N$ evaluations of the objective function $f(\cdot)$, it implies that the computational complexity of {\tt LS} is $O(\frac{1}{\epsilon}N^3 \log N)$ in the worst case.
		Despite calling the procedure {\tt LS} twice, the overall computational complexity of {\tt UM} (Algorithm~\ref{algo:UM_algorithm}) is also $O(\frac{1}{\epsilon} N^3 \log N)$, which is polynomial in $N$.
	\end{proof}

\end{appendix}

\end{document}